\documentclass [11pt] {article} 

\usepackage{fullpage}
\usepackage[pdftex]{graphicx}

\usepackage{color}
\usepackage{amsmath}
\usepackage{amssymb}
\usepackage{amsfonts}
\usepackage{enumitem}

\usepackage[english]{babel}

\usepackage{multirow}
\usepackage{rotating}

\usepackage{epstopdf}

\begin{document}

\newtheorem{theorem}{Theorem}[section]
\newtheorem{lemma}{Lemma}[section]
\newtheorem{corollary}{Corollary}[section]
\newtheorem{claim}{Claim}[section]
\newtheorem{proposition}{Proposition}[section]
\newtheorem{definition}{Definition}[section]
\newtheorem{fact}{Fact}[section]
\newtheorem{example}{Example}[section]

\newcommand{\cA}{{\cal A}}
\newcommand{\cP}{{\cal P}}
\newcommand{\cC}{{\cal C}}
\newcommand{\cG}{{\cal G}}
\newcommand{\cN}{{\cal N}}
\newcommand{\cU}{{\cal U}}
\newcommand{\cT}{{\cal T}}
\newcommand{\cS}{{\cal S}}
\newcommand{\cL}{{\cal L}}
\newcommand{\cV}{{\cal V}}
\newcommand{\loc}{{\cal LOCAL}}

\newcommand{\f}[4][0pt]{%
\begin{list}{#2}{%
 \setlength{\leftmargin}{#3 em}\addtolength{\leftmargin}{#3 em}\addtolength{\leftmargin}{1 em}%
 \setlength{\labelsep}{#3 em}\addtolength{\labelsep}{#3 em}
 \setlength{\labelwidth}{20pt}
 \if!#1!\else\addtolength{\leftmargin}{#1}\fi%
 \setlength{\topsep}{2pt}% spaces between lines of code
 \setlength{\partopsep}{0pt}%
 \if!#1!\else\setlength{\itemindent}{-#1}\fi%
}%
   \item  #4%
\end{list}%
}

\newcommand{\qed}{\hfill $\square$ \smallbreak}
\newenvironment{proof}[1][Proof]
{\par\noindent{\bf #1:} }{\hspace*{\fill}\nolinebreak{$\Box$}\bigskip\par}

\newcommand{\view}{\mathcal{V}}
\newcommand{\agent}{\lambda}
\newcommand{\lab}{\alpha}
\newcommand{\labels}{\mathcal{L}}
\newcommand{\lista}{\mathcal{Q}}
\newcommand{\code}{\xi}
\newcommand{\home}{h}
\newcommand{\routebegin}{b}
\newcommand{\routeend}{d}
\newcommand{\cR}{\mathcal{R}}
\newcommand{\hist}{\mathcal{H}}
\newcommand{\cQ}{\mathcal{Q}}

\newcommand{\ints}{\mathbb{N}}
\newcommand{\algorithmCL}{{\tt Choose\textup{-}Leader}}
\newcommand{\algorithmUL}{{\tt Update\textup{-}Label}}
\newcommand{\algorithmLE}{{\tt Leader\textup{-}Election}}
\newcommand{\algorithmInit}{{\tt Initialization}}

\newcommand{\last}{\textup{last}}
\newcommand{\depth}{3(n-1)}
\newcommand{\len}{\ell}
\newcommand{\conditionEC}{\textup{EC}}

\title{{\bf  Leader Election for Anonymous Asynchronous Agents in Arbitrary Networks}}

\author{Dariusz Dereniowski\thanks{Gdansk University of Technology,
ETI Faculty,
Department of Algorithms and System Modeling,
ul. Narutowicza 11/12,
80-233 Gda\'{n}sk, Poland {\tt deren@eti.pg.gda.pl}.
This work was done during the visit of Dariusz Dereniowski at the
Research Chair in Distributed Computing of the Universit\'e du
Qu\'ebec en Outaouais.
This author was partially supported by the Polish Ministry of Science and Higher Education (MNiSW) grant N~N516~196437.}
 \and Andrzej Pelc\thanks{
 D\'epartement d'informatique, Universit\'e du Qu\'ebec en Outaouais, Gatineau,
Qu\'ebec J8X 3X7, Canada. {\tt pelc@uqo.ca}. Partially supported by NSERC discovery grant 
and by the Research Chair in Distributed Computing at the
Universit\'e du Qu\'ebec en Outaouais.}
}

\maketitle

\thispagestyle{empty}

\begin{abstract}

We study the problem of leader election among mobile agents operating in an arbitrary network modeled as an undirected graph.
Nodes of the network are unlabeled and all agents are identical. Hence the only way to elect a leader among agents is by exploiting asymmetries
in their initial positions in the graph. Agents do not know the graph or their positions in it, hence they must gain this knowledge
by navigating in the graph and share it with
other agents to accomplish leader election. This can be done using meetings of agents, which is difficult because of their asynchronous nature: 
an adversary has total control over the speed of agents. When can a leader be elected in this adversarial scenario and how to do it?
We give a complete answer to this question by characterizing all initial configurations for which leader election is possible and by constructing an algorithm
that accomplishes leader election for all configurations for which this can be done.

\vspace*{1cm}

{\bf keywords:} leader election,  anonymous network, asynchronous mobile agents

\vspace*{5cm}
\end{abstract}

\pagebreak
%%%%%%%%%%%%%%%%%%%%%%%%%%%%%%%%%%%%
\section{Introduction} \label{sec:intro}
%%%%%%%%%%%%%%%%%%%%%%%%%%%%%%%%%%%%

\subsection{The model and the problem} \label{subsec:model}

Leader election is one of the fundamental problems in distributed computing, first stated in \cite{LL}. Each entity in some set
 has a Boolean variable initialized to 0 and, after the election, exactly one
of these entities, called the {\em leader}, should change this value to 1. All other entities should know which one becomes the leader.
In this paper we consider the problem of leader election among mobile agents that operate in a network.
We assume that neither nodes of the network nor agents have labels that can be used for leader election. This assumption is motivated by scenarios
where nodes and/or agents may refrain from revealing their identities, e.g., for security or privacy reasons.
Hence it is desirable to have leader election algorithms that do not rely on identities but exploit
asymmetries in the initial configuration of agents due to its topology and to port labelings. 
With unlabeled nodes and agents, leader election is impossible for symmetric initial configurations,
e.g., in a ring  in which ports at each node are 0,1, in the clockwise direction and agents are initially situated at every node. 
Our goal is to answer the following question:
\begin{quotation}
For which initial configurations of agents is leader election possible and how to do it when it is possible?
\end{quotation}

A network is modeled as an undirected connected graph with unlabeled nodes.
It is important to note that the agents have to be able to {\em locally} distinguish ports at a
node: otherwise, the adversary could prevent an agent from choosing a particular edge, thus 
making navigation in the network impossible even in the simple case of trees. 
This justifies a common 
assumption made in the literature:
ports at a node of degree $d$ have arbitrary fixed labelings $0,\dots,$ $d-1$.
Throughout the paper, we will use the term ``graph'' to mean a graph with the above properties.  
We do not assume any coherence between port labelings at various nodes.
Agents can read the port numbers when entering and leaving nodes.

At the beginning, identical agents are situated in some nodes of the graph, at most one agent at each node.
The graph with bicolored nodes (black if the node is occupied, white if it is not) is called an initial configuration.
Agents  do not have labels and have unlimited memory: they are modeled as identical Turing machines.
They execute the same deterministic algorithm.
%An agent that enters a node, reads the entry port number and the degree of the node. On the basis of its memory
%it chooses the exit port number or decides to stay at a given node.

Agents navigate in the graph in an asynchronous way which is formalized by an adversarial model used in 
\cite{BCGIL,CCGL,CLP,DGKKPV,GP} and described below.
Two important notions used to specify movements of agents are the {\em route} of the agent and its {\em walk}.
Intuitively, the agent chooses the route {\em where} it moves and the adversary describes the walk on this 
route, deciding {\em how} the agent  moves. More precisely,  these notions are defined as follows.
The adversary initially places an agent at some node of the graph.
The route is chosen by the agent and is defined as follows. 
The agent chooses one of the available ports at the current node. 
After getting to the other end of the corresponding edge, the agent learns the port number by 
which it enters and the degree of the entered node. Then it chooses one of the available ports at this node
or decides to stay at this node.
The resulting route of the agent is the corresponding sequence of edges $(\{v_0,v_1\},\{v_1,v_2\},\dots)$,
which is a (not necessarily simple) path in the graph. 

We now describe the walk $f$ of an agent on its route. Let $R=(e_1,e_2,\dots)$ be the route of an agent. Let $e_i=\{v_{i-1},v_i\}$.
Let $(t_0,t_1,t_2,\dots)$, where $t_0=0$, be an increasing sequence of reals, chosen by the adversary, 
that represent points in time. Let $f_i:[t_i,t_{i+1}]\rightarrow [v_i,v_{i+1}]$ be any continuous function, chosen by the adversary, such that $f_i(t_i)=v_i$ and $f_i(t_{i+1})=v_{i+1}$. For any $t\in [t_i,t_{i+1}]$, we define $f(t)=f_i(t)$. 
The interpretation of the walk $f$ is as follows: at time $t$ the agent
is at the point $f(t)$ of its route.  This general definition of the walk and the fact that (as opposed to the route) it is designed by the adversary,
are a way to formalize the asynchronous characteristics of the process.  The movement of the agent can be
at arbitrary speed, the adversary may sometimes stop the agent or move it back and forth, as long as the walk 
in each edge of the route is continuous and covers all of it.
%\footnote{In \cite{BCGIL,CCGL,CLP,DGKKPV,GP} the walk was defined 
%in a slightly more general way, without assuming that the functions $f_i$ are non-decreasing.
%This corresponds to the additional possibility of the adversary to move 
%the agent back and forth on an edge. We eliminate this possibility for simplicity
%of presentation of the model, regarding the memory states of the agents 
%defined in Section 2, but all our results remain valid for the original definition as well. }
This definition makes the adversary very powerful,
and consequently agents have little control on how they move. This, for example, makes meetings between agents hard to achieve.
Note that agents can meet either at nodes or inside edges of the graph. 

Agents with routes $R_1$ and $R_2$ and with walks $f_1$ and $f_2$ meet at time $t$,
if points $f_1(t )$ and $f_2(t )$ are identical. A meeting is guaranteed for routes $R_1$ and $R_2$,
if the agents using these routes meet at some time $t$, regardless of the walks chosen by the adversary.
%A rendezvous algorithm executed by an agent in a graph produces the route of the agent, given its starting point
%(and results of coin tosses in the randomized scenario). We say that asynchronous rendezvous is {\em feasible} from
%given initial positions, if there exist routes $R_1$ and $R_2$ starting from these positions that guarantee rendezvous.
When agents meet, they notice this fact and can exchange all previously acquired information. However, if the meeting is inside an edge,
they continue the walk prescribed by the adversary until reaching the other end of the current edge. New knowledge acquired at the meeting
can then influence the choice of the subsequent part of the routes constructed by each of the agents.

Since agents do not know a priori the topology of the graph and have identical memories at the beginning,
the only way to elect a leader among agents is by learning the asymmetries
in their initial positions in the graph. Hence agents must gain this knowledge by navigating in the network and share it with
other agents to accomplish leader election. Sharing the knowledge can be done only as a result of meetings of agents, 
which is difficult because of the asynchronous way in which they move.

It is not hard to see (cf. Proposition \ref{no-bound})  that in the absence of a known upper bound on the size of the graph, leader election is impossible even for asymmetric configurations.  Hence we assume that all agents know a priori a common upper bound $n$ on the size of the graph. This is the only information
about the environment available to the agents when they start the task of leader election.

Having described our model, we can now make the initial problem more precise. Call an initial configuration {\em eligible} if, starting from this configuration,
leader election can be accomplished regardless of the actions of the adversary. Thus in order that a  configuration be eligible, it is
enough to have some leader election algorithm starting from it, even one dedicated to this specific configuration.  Now our problem can be reformulated as follows.

\begin{quotation}
Which initial configurations are eligible?
Find a universal leader election algorithm that  elects a leader regardless of the actions of the adversary, for all eligible configurations in graphs of size at most $n$, where $n$ is known to the agents.
\end{quotation}

%%%%%%%%%%%%%%%%%%%%%%%%%%%%%%%%%%%%
\subsection{Our results} \label{subsec:our_results}
%%%%%%%%%%%%%%%%%%%%%%%%%%%%%%%%%%%%
Assuming an upper bound $n$ on the size of the graph, known a priori to all agents,
we characterize all eligible initial configurations and construct an algorithm
that accomplishes leader election for all of them. More precisely, we formulate a combinatorial condition
on the initial configuration, which has the following properties. On the one hand,  if this condition does not hold, 
then the adversary can prevent leader election starting from the given initial configuration. 
On the other hand, we construct an algorithm that elects a leader, regardless of the adversary, for all initial
configurations satisfying the condition, in graphs of size at most equal to the given bound $n$.

Intuitively, leader election is possible when the initial configuration is asymmetric and when agents can learn this, regardless
of the actions of the adversary.
Both these requirements are contained in the necessary and sufficient condition on eligibility, which we formulate in Section~\ref{sec:feasibility}. In fact, the process of learning the 
asymmetries by the agents is the main conceptual and  technical challenge in the design and analysis of our algorithm.
Agents acquire and share this knowledge as a result of meetings. The difficulty is to design the algorithm in such a way that all asymmetries be finally learned
by all agents and that  all agents be aware of  this fact  and thus capable to correctly elect the leader. 
%%%%%%%%%%%%%%%%%%%%%%%%%%%%%%%%%%%%
\subsection{Related work} \label{subsec:related_work}
%%%%%%%%%%%%%%%%%%%%%%%%%%%%%%%%%%%%
Leader election in networks was mostly studied assuming that all nodes have distinct labels and election has to be performed among nodes.
This task was first studied for rings.
A synchronous algorithm, based on comparisons of labels, and using
$O(n \log n)$ messages was given in \cite{HS}. It was proved in \cite{FL} that
this complexity is optimal for comparison-based algorithms. On the other hand, the authors showed
an algorithm using a linear number of messages but requiring very large running time.
An asynchronous algorithm using $O(n \log n)$ messages was given, e.g., in \cite{P} and
the optimality of this message complexity was shown in \cite{B}. Deterministic leader election in radio networks has been studied, e.g., 
in \cite{JKZ,KP,NO} and randomized leader election, e.g., in \cite{Wil}. In \cite{HKMMJ} the leader election problem is
approached in a model based on mobile agents for networks with labeled nodes.

Many authors \cite{ASW,AtSn,BV,DKMP,Kr,KKV,Saka,YK,YK3} studied various computing
problems in anonymous networks. In particular, \cite{BSVCGS,YK3} characterize message passing networks in which
leader election can be achieved when nodes are anonymous. In \cite{YK2} the authors study
the problem of leader election in general networks, under the assumption that labels are
not unique. They characterize networks in which this can be done and give an algorithm
which performs election when it is feasible. They assume that the number of nodes of the
network is known to all nodes. In
 \cite{FKKLS}  the authors
study feasibility and message complexity of sorting and leader election in rings with
nonunique labels, while in \cite{DP} the authors provide algorithms for the
generalized leader election problem in rings with arbitrary labels,
unknown (and arbitrary) size of the ring and for both
synchronous and asynchronous communication. 
Characterizations of feasible instances for leader election and naming problems have been provided in~\cite{C,CMM,CM}.
Memory needed for leader election in unlabeled networks has been studied in \cite{FP}. 

The asynchronous model for mobile agents navigation in unlabeled networks has been previously used in \cite{BCGIL,CCGL,CLP,DGKKPV,GP} 
in the context of rendezvous between two  agents. In \cite{BCGIL,CCGL,CLP,DGKKPV} agents had different labels and in \cite{GP} agents were
anonymous, as in our present scenario. The synchronous model, in which agents traverse edges in lock-step, has been used, e.g., in \cite{CKP,DiPe,TSZ07},
also in the context of rendezvous.

\subsection{Roadmap} \label{subsec:roadmap}

In Section~\ref{sec:memory_states} we formalize the description of how agents decide. This concerns both navigation decisions (on what basis the agents construct their routes) and the final decision who is the leader. We define memory states of the agents that are the basis of all these decisions. In Section~\ref{sec:feasibility} we formulate the combinatorial condition EC
concerning initial configurations that is then proved to be equivalent to eligibility, and we formulate our main result. In Section~\ref{sec:negative} we prove two negative results concerning leader election:
one saying that condition EC is necessary for eligibility and the other saying that the assumption concerning knowledge of the upper bound cannot be removed.
In Section~\ref{sec:algorithm} we give our main contribution: we construct a universal algorithm electing a leader for all configurations satisfying condition EC, if agents know
an upper bound on the size of the graph. Section~\ref{sec:conclusions} contains conclusions.

%%%%%%%%%%%%%%%%%%%%%%%%%%%%%%%%%%%%
\section{Memory states and decisions of agents} \label{sec:memory_states}
%%%%%%%%%%%%%%%%%%%%%%%%%%%%%%%%%%%%

In this section we describe formally on what basis the agents make decisions concerning navigation in the graph (i.e., how they construct their routes)
and on what basis they make the decision concerning leader election. All these decisions depend on the {\em memory states} of the agents.
At every time $t$ the memory state of an agent
is a finite sequence of symbols defined as follows. Before an agent is woken up by the adversary, its memory state is blank: it is the empty sequence.
When an agent is woken up, it perceives the degree $d$ of its initial position, i.e., its memory state becomes the sequence $(d)$. Further on, the memory state of an
agent changes when it visits a node. It is caused by 
the following three types of events: entering a node by the agent, meeting other agents, and leaving a node by the agent.
A change of a memory state of an agent is done by appending to its current memory state a sequence of symbols defined as follows. The change due to
entering a node of degree
$d$ by port number $p$, consists of appending the sequence $(p,d)$ to the current memory state of the agent. The change due to leaving a node by port $q$  
consists of appending $q$
to the current memory state of the agent. The change due to meeting other agents is defined as follows.
When entering a node $v$ the agent considers all meetings with other agents that occurred since leaving the previous node.
Suppose that the current memory states of the agents met in this time interval by agent $\lambda$ were $\sigma_1,\dots, \sigma_k$, in lexicographic order, regardless of the order of meetings
in this time interval and disregarding repeated meetings corresponding to the same memory state (and thus to the same agent). 
Agent $\lambda$ appends the sequence of symbols $([\sigma_1]...[\sigma_k])$
 to its current memory state. When two or more of these events occur simultaneously, for example an agent meets another agent when it enters
a node, or an agent meets simultaneously several agents, then the appropriate sequences are appended to its current memory one after another, in lexicographic order.
When in the previous memory state the agent made a decision to stay idle at the current node, then its memory state can change only if and when some other agent enters this node.
This completes the description of how the memory states of agents evolve. 
Notice that after traversing an edge the action of agent $\lambda$ consisting of appending a sequence of symbols $[\sigma]$
due to a meeting with an agent with current memory state $\sigma$ since leaving the previous node, is performed by $\lambda$ at most once. 
Since the number of agents is finite, this implies that, by any given moment in time, the memory state of an agent has changed only a finite number of times, 
and each time a finite sequence of symbols has been  appended. Hence memory
states are indeed finite sequences of symbols.

The decisions of agents are made always when an agent is at a node and they are of three possible types: an agent can decide to stay idle, it can decide to exit
the current node by some port, or it can elect a leader and stop forever. All these decisions are based on the memory state of the agent after entering the current node and are 
prescribed by the algorithm. (Recall that agents execute the same deterministic algorithm.) If an agent decides to stay at a given node, then it remains idle at it until
another agent  enters this node. At this time the memory state of the idle agent changes, and in the new memory state the agent makes
a new decision. If an agent decides to leave the current node by a given port, it walks in the edge in the way prescribed by the adversary and makes a new decision
after arriving at the other end of the edge. Finally, if an agent decides to elect a leader, it either elects itself, or it decides that it is not a leader, in which case
it has to give a sequence of port numbers leading from its own initial position to the initial position of the leader: 
this is the meaning of the requirement that every non-leader has
to know which agent is the leader.

%%%%%%%%%%%%%%%%%%%%%%%%%%%%%%%%%%%%
\section{Feasibility of leader election} \label{sec:feasibility}
%%%%%%%%%%%%%%%%%%%%%%%%%%%%%%%%%%%%

In this section we express the necessary and sufficient condition on eligibility of an initial configuration and we formulate the main result of this paper.
We first introduce some basic terminology.

We will use the following notion from \cite{YK3}. Let $G$ be a graph and $v$ a node of $G$.  We first define, for any $l \geq 0$,  the {\em truncated view}
$\cV^l(v)$ at depth $l$, by induction on $l$. $\cV^0(v)$ is a tree consisting of a single node $x_0$. 
If $\cV^l(u)$ is defined for any node $u$ in the graph, then $\cV^{l+1}(v)$ is the port-labeled tree
rooted at $x_0$ and defined as follows.
For every node $v_i$, $i=1,\dots ,k$, adjacent to $v$, 
there is a child $x_i$ of $x_0$ in $\cV^{l+1}(v)$ such that the port number at $v$ corresponding to edge $\{v,v_i\} $ is the same as the port number 
at $x_0$ corresponding to edge $\{x_0,x_i\}$,
and the port number at $v_i$ corresponding to edge $\{v,v_i\} $ is the same as the port number at $x_i$ corresponding to edge $\{x_0,x_i\}$. 
Now node $x_i$, for $i=1,\dots ,k$ becomes 
the root of the truncated view $\cV^l(v_i)$.

The {\em view} from $v$ is the infinite rooted tree $\cV(v)$ with labeled ports, such that $\cV^l(v)$ is its truncation to level $l$, for each $l\geq 0$.
For an initial configuration in which node $v$ is the initial position of an agent, the view $\cV(v)$ is called the view of this agent.

We will also use a notion similar to that of the view but reflecting the positions of agents in an initial configuration.
Consider a graph $G$ and an initial configuration of agents in this graph. Let $v$ be a node occupied by an agent. 
A function $f$ that assigns either $0$ or $1$ to each node of
$\cV(v)$ is called a \emph{binary mapping} for $\cV(v)$.
A pair $(\cV(v),f)$, where $f$ is a binary mapping for $\cV(v)$, such
that $f(x)=1$ if and only if $x$ corresponds to an initial position of an
agent, is called the \emph{enhanced view} from $v$.
Thus, the enhanced view of an agent additionally marks in its view the nodes corresponding to initial positions of other agents in the initial configuration.

For any route $R=(e_1,e_2,\dots, e_k)$ such that $e_i=\{v_{i-1},v_i\}$, we denote $b(R)=v_0$ and $d(R)=v_k$,
and we say that $R$  \emph{leads} from $v_0$ to $v_k$ in $G$.
Since nodes of $G$ are unlabeled, agents traveling on a route are aware only of the port numbers of the edges they traverse.
Hence, it will be usually more convenient  to refer to these sequences of port numbers rather than to the edges of the route.
Any finite sequence of non-negative integers will be called a \emph{trail}.

We define an operator $\cT$, that provides the trail corresponding to a given route.
More formally, if $R=(\{v_1,v_2\},\{v_2,v_3\},\ldots,\{v_{j-1},v_j\})$ is a route in $G$, then define $\cT(R)=(p_1,\ldots,p_{2j-2})$ to be the trail such that $p_{2i-1}$ and $p_{2i}$ are the port numbers of $\{v_i,v_{i+1}\}$ at $v_i$ and $v_{i+1}$, respectively, for $i=1,\ldots,j-1$.
We say that a trail $T$ is \emph{feasible from} $v$ in $G$, if there exists a route $R$ in $G$ such that $\routebegin(R)=v$ and $\cT(R)=T$, and in such a case the route $R$ is denoted by $\cR(v,T)$. 

For a sequence $A=(a_1,\ldots,a_k)$ we denote by $\overline{A}$ the sequence $(a_k,a_{k-1},\dots,a_1)$.
For two sequences $A=(a_1,\ldots,a_k)$ and $B=(b_1,\ldots,b_r)$ we write $(A,B)$ to refer to the sequence $(a_1,\ldots,a_k,b_1,\ldots,b_r)$.

For any agent $\agent$, let $\home(\agent)$ denote its initial position.
Consider two agents $\agent$ and $\agent'$.
Consider any route $R$ leading from $\home(\agent)$ to $\home(\agent')$ and let $T=\cT(R)$. 
If $T=\overline{T}$, then we say that the route $R$ is a palindrome. For a given initial configuration, a palindrome $R$ is called {\em uniform},
if for any route $R'$ such that $\cT(R')=\cT(R)$, whenever $b(R')$ is occupied by an agent, then $d(R')$ is also occupied by an agent.

We are now ready to formulate our condition on an initial configuration, that will be called EC (for eligibility condition) in the sequel:

\begin{center}
Enhanced views of all  agents are different {\em and}\\ (There exist agents with different views {\em or} There exists a non-uniform palindrome)
\end{center}

We now formulate our main result whose proof is the objective of the rest of the paper.

\begin{theorem}\label{main}
Assume that all agents are provided with an upper bound $n$ on the size of the graph.
Then an initial configuration is eligible  if and only if condition {\em CE} holds for this configuration. Moreover, there exists
an algorithm electing a leader for all eligible configurations, regardless of the actions of the adversary.
\end{theorem}

\section{The negative results} \label{sec:negative}

In this section we prove two negative results concerning the feasibility of leader election. The first result shows that
condition CE is necessary to carry out leader election, even if the graph (and hence its size) is known to the agents.

\begin{proposition} \label{pro:negative1}
Suppose that the condition {\em CE} does not hold for the initial configuration. Then there exists an adversary, such that
leader election cannot be accomplished for this configuration, even if the graph is known to the agents.
\end{proposition}
\begin{proof}
Fix an initial configuration.
Condition CE can be abbreviated as $\alpha \wedge (\beta \vee \gamma)$, where
$\alpha$ is ``Enhanced views of all  agents are different'', $\beta$ is ``There exist agents with different views'', and $\gamma$ is ``There exists a non-uniform palindrome''. Suppose that EC does not hold. 

First consider the case when $\alpha$ is false.
This means that there exist agents $\lambda$ and $\lambda'$ with the same enhanced view. This in turn
implies that for every agent $\mu$ there exists an agent $\mu'$ that has the same enhanced view as $\mu$. 
Indeed, if $T$ is the trail corresponding to a route that leads from $\home(\lambda)$ to $\home(\mu)$, agent $\mu'$ is the agent whose initial position is at the end of the route corresponding to $T$ and starting at $\home(\lambda')$. For every agent $\mu$ we will call the agent $\mu'$ its twin.
Consider a hypothetical
leader election algorithm and the
``perfectly synchronous'' adversary that starts the execution of the algorithm simultaneously for all agents and 
moves all of them with the same constant speed. Such an adversary induces rounds which are units of time in which
all agents traverse an edge. The beginning of a round coincides with the end of the previous round. 
Hence at the beginning and at the end of each round every agent is at a node. If agents meet inside an edge, they must
meet exactly in the middle of a round in which they traverse an edge in opposite directions.
We will show that the memory state of twins is identical at the end of each round. This implies that leader election is impossible, as an agent elects a leader when it is at a node, and consequently if some agent elects itself as a leader, its twin would elect itself as well,
violating the uniqueness of the leader. 

The invariant that the memory state of twins is identical at the end of each round 
is proved by induction on the round number. It holds at the beginning, due  to the same degree of initial positions of twins.
Suppose that after some round $i$ the memory  states of twins are identical. Consider twins $\mu$ and $\mu'$.
In round $i+1$ they exit by the same port number and enter the next node by the same port number. If in round $i+1$
they don't meet any agent in the middle of the edge, at the end of the round $\mu$ must meet agents with the same
memory states as those met by $\mu'$ (if any), and hence memory states of $\mu$
and $\mu'$ at the end of round $i+1$ are identical. If in round $i+1$ agent $\mu$ meets some agents in the middle of the edge,
then agent $\mu'$ must meet exactly the twins of these agents in the middle of the edge.
By the inductive hypothesis, these twins have the same memory states as agents met by $\mu$ and hence again, at the end of the round the memory states of $\mu$ and $\mu'$ are identical.
This concludes the proof if $\alpha$ is false.

Next consider the case when $\beta \vee \gamma$ is false. This means that views of all agents are identical and 
every palindrome for the initial configuration (if any) is uniform. For any trail $\pi$ that yields a uniform palindrome, 
this gives a partition
of all agents into pairs $(\mu_{\pi}, \mu'_{\pi})$ of agents at the ends of routes that correspond to this trail.

Again we consider the ``perfectly synchronous'' adversary described above. There are two subcases. If there is no palindrome
in the initial configuration, then we prove the following invariant, holding at the beginning of each round,  
by induction on the round number: 
the memory state of all agents is the same and there is no palindrome between agents. The invariant holds at the beginning
by assumption. Suppose it holds after round $i$. In round $i+1$ all agents choose the same port number and enter the
next node by the same port number. There are no meetings in round $i+1$. 
Indeed, the only meeting could be in the middle of an edge but this would mean that agents were joined by a
one-edge palindrome at the beginning of round $i+1$. If a pair of agents were joined by a palindrome after round $i+1$, they
would have to be joined by a palindrome longer or shorter by two edges at the beginning of the round, contradicting the inductive assumption. Hence the invariant holds by induction.

The second subcase is when there is a palindrome in the initial configuration (and hence all such palindromes are uniform). 
Now we prove the following invariant
holding in the beginning of each round:
 the memory state of all agents is the same and every agent is at the end of a palindrome corresponding to the same trail.
The invariant holds at the beginning by the assumption.
 Suppose the invariant holds after round $i$. In round $i+1$ all agents choose the same port number and enter the
next node by the same port number. If after round $i$ no pair of agents were at the ends of an edge with both ports equal,
or they were but agents did not choose this port in round $i+1$, then no meeting occurred and the invariant carries on after round $i+1$.
If after round $i$ every pair of agents were at the ends of an edge with both ports $p$ and the agents chose this port
in round $i+1$, then meetings of agents with identical memory states occurred in pairs in the middle of each joining edge.
Since meeting agents had identical memory state during the meeting, this holds also after round $i+1$ and agents are again
in pairs at the ends of edges with both ports $p$. Thus the invariant holds at the end of round $i+1$.

Hence at the beginning of each round the memory state of all agents is the same, both when $\alpha$ and when
$\beta \vee \gamma$ is violated. This implies that with the ``perfectly synchronous'' adversary leader election is impossible
whenever condition EC is violated for the initial configuration. Notice that the argument holds even when agents know the graph in which they operate.
\end{proof}

Our second negative result shows that the assumption about the knowledge of an upper bound on the size of the graph
cannot be removed from Theorem \ref{main}.

\begin{proposition}\label{no-bound}
There is no algorithm that accomplishes leader election regardless of the adversary for all initial configurations satisfying condition {\em CE}. 
\end{proposition}
\begin{proof}
Suppose for a contradiction that such a universal algorithm ${\tt A}$ exists. Consider an  ``almost'' oriented ring of size $m$: ports 0,1 are in the clockwise direction at each node except one, where they are counterclockwise. This node is called {\em special}.
The initial configuration on this ring consists of two agents: one at the special node, and one at the neighbor clockwise from it.
Call this configuration $C_1$ (cf. Fig. \ref {fig:negative} (a)).
This configuration satisfies condition CE: agents have different views. Hence algorithm ${\tt A}$ must elect a leader for this configuration, regardless of the adversary. Consider a ``perfectly synchronous'' adversary that starts the execution of the algorithm simultaneously for all agents and 
moves all of them with the same constant speed. It induces rounds corresponding to edge traversals by all agents.
Suppose that a leader is elected for this adversary after $t$ rounds.

Now consider a ring of size $4tm$ in which there are $4t$ special nodes at distances $m$: at these nodes ports 0,1
are in the counterclockwise direction, and in all other nodes they are in the clockwise direction. The initial configuration
consists of $8t+1$ agents. There is an agent at every special node and at every clockwise neighbor of a special node.
Additionally there is an agent at the counterclockwise neighbor of one special node. Call this configuration $C_2$ (cf. Fig. \ref {fig:negative} (b)). 
This configuration satisfies condition
CE. Indeed, due to the single group of three consecutive agents, all agents have distinct enhanced views. On the other hand,
agents at special nodes have a different view from agents at clockwise neighbors of special nodes. Hence algorithm ${\tt A}$ must elect a leader for this configuration as well, regardless of the adversary. Consider the same  ``perfectly synchronous'' adversary
as before.

\begin{figure}[hbt]
	\begin{center}
	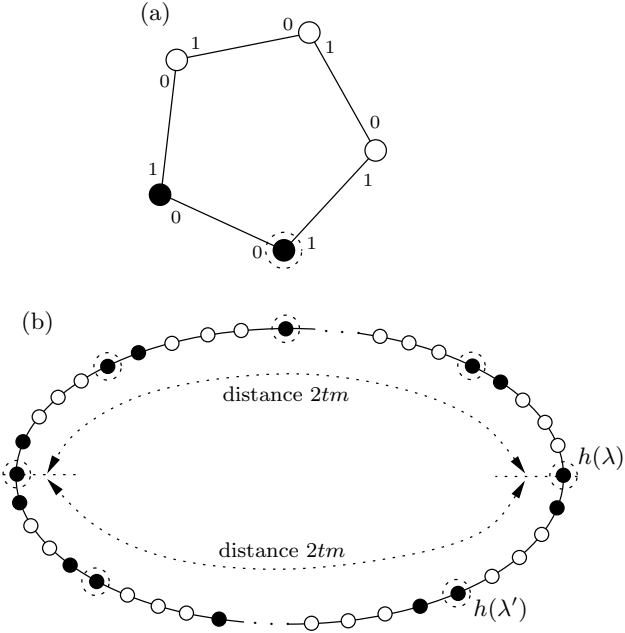
	\caption{(a) configuration $C_1$ in a ring of size $m=5$;
                 (b) configuration $C_2$ in the corresponding ring of size $4tm$. Special nodes are encircled.}
	\label{fig:negative}
	\end{center}
\end{figure}

 Consider the
agent $\lambda$ in the configuration $C_2$ that is initially situated at the special node $v$ antipodal 
to the special node with both neighbors hosting agents. Consider the agent $\lambda'$ initially situated at the
special node $v'$ that is clockwise from $v$ and closest to $v$.
In the first $t$ rounds of the execution of ${\tt A}$ starting from configuration $C_2$, memory states of agents $\lambda$ and $\lambda'$ are identical to memory states of the agent $\mu$ initially situated at the special node of configuration $C_1$. This is due to the large size of the ring in configuration $C_2$.
Hence if in configuration $C_1$ agent $\mu$ elects itself, then in configuration $C_2$ agents $\lambda$ and $\lambda'$ elect 
each of them itself as the leader after $t$ rounds. If in configuration $C_1$ agent $\mu$ elects its neighbor,  then in configuration 
$C_2$ agents $\lambda$ and $\lambda'$ elect each of them their neighbor as the leader after $t$ rounds. In both cases two different agents are elected, which is a contradiction.
\end{proof}

\section{The algorithm and its correctness} \label{sec:algorithm}

In this section we present an algorithm that elects a leader for all initial configurations satisfying condition EC, assuming that an upper bound on the size of the graph
is known to all agents. In view of Proposition \ref{no-bound}, this assumption cannot be removed. This upper bound, denoted by $n$, is an input of our algorithm.

The section is divided into three subsections. In the first subsection we provide additional terminology and notation, as well as some auxiliary results used in the algorithm and in its analysis.
In the second  subsection we give the intuitive overview of the algorithm and its formal description, and we provide some illustrative examples of its functioning.
Finally, the third subsection is devoted to the proof that the algorithm is correct.

\subsection{Additional notions and auxiliary results} \label{subsec:alg:defs}

Let $G$ be a graph and let $v$ be any node of $G$. For any integer $l>0$, we define the \emph{code} of $\view^l(v)$ as the sequence $\code(\view^l(v))$ of integers obtained as follows. Perform the complete depth first search traversal of $\view^l(v)$, starting at its root,  in such a way that at each node an edge with a smaller port number is selected prior to an edge with a larger port number. Then, if $\{x,y\}$ is the $i$-th traversed edge, when going from $x$ to $y$ during the traversal of $\{x,y\}$, then the $(2i-1)$-th and $(2i)$-th elements of $\code(\view^l(v))$ are the port numbers of $\{x,y\}$ at $x$ and $y$, respectively. The following is a direct consequence of this definition.
\begin{proposition} \label{prop:codes_distinguish}
Let $u$ and $v$ be any nodes of $G$ and let $l>0$ be an integer. $\view^l(u)\neq\view^l(v)$ if and only if $\code(\view^l(u))\neq\code(\view^l(v))$.
\qed
\end{proposition}

Let $l>0$ be an integer.
We extend the notion of binary mappings to the truncated views.
We say that $f$ is a \emph{binary mapping} for $\view^l(v)$, if $f$ assigns either $0$ or $1$ to each node of $\view^l(v)$.
If $f$ is a binary mapping for $\view(v)$ (or for $\view^{l'}(v)$ for some $l'>0$), then $(\view^l(v),f)$ (where $l\leq l'$, respectively) refers to $f$ restricted to the nodes of $\view^l(v)$.
Given two binary mappings $f_1,f_2$ for $\view^l(v)$, we write $f_1\trianglelefteq f_2$ if $f_1(x)\leq f_2(x)$ for each node $x$ of $\view^l(v)$.
If $(\view(v),f)$ is the enhanced view from $v$ and $f'$ is a binary mapping for $\view^l(v)$ such that $f'\trianglelefteq f$, then the pair $(\view^l(v),f')$ is called a \emph{partially enhanced view} from $v$. Intuitively, in a partially enhanced view only some nodes corresponding to initial positions of agents are marked.

Let $(\view(v),f)$ be the enhanced view from $v$, where $v$ is selected so that there exists an agent $\agent$ with $\home(\agent)=v$.
Then, $(\code(\view^{n-1}(v)),f)$ is called the \emph{complete identifier} of $\agent$.
The significance of the notion of a complete identifier is the following.
An agent can never get the entire view or the entire enhanced view, as these are infinite objects.
However, the following propositions from \cite{Norris} show that to differentiate two views or two enhanced views, it is enough to consider their truncations to depth $n-1$.
Thus, as stated in Corollary~$\ref{cor:No2}$, complete identifiers identify agents with different enhanced views.

\begin{proposition} \textup{(\cite{Norris})} \label{prop:No}
For a $n$-node graph $G$ and for all nodes $u$ and $v$ of $G$,
$\view(u)=\view(v)$ if and only if $\view^{n-1}(u)=\view^{n-1}(v)$.
\qed
\end{proposition}
\begin{proposition} \textup{(\cite{Norris})} \label{prop:No2}
For a $n$-node graph $G$, for all nodes $u$ and $v$ of $G$, if $(\view(u),f)$ and $(\view(v),f')$ are the enhanced views from $u$ and $v$, respectively, then 
$(\view(u),f)=(\view(v),f')$ if and only if $(\view^{n-1}(u),f)=(\view^{n-1}(v),f')$.
\qed 
\end{proposition}
\begin{corollary} \label{cor:No2}
For a $n$-node graph $G$ and for any agents $\agent$ and $\agent'$,
the enhanced views from $\home(\agent)$ and $\home(\agent')$ are equal if and only if the complete identifiers of $\agent$ and $\agent'$ are equal.
\qed
\end{corollary}

A sequence of the form $\lab=(\code(\view^l(v)),f_1,f_2,\ldots,f_j)$ is called a \emph{label}, if the following conditions hold:
\begin{enumerate}[label={\normalfont(\roman*)}]
 \item $v$ is a node of $G$, and $l>0$ and $j>0$ are integers,
 \item $f_i$ is a binary mapping for $\view^l(v)$ for each $i=1,\ldots,j$.
 \item $f_1$ is the binary mapping for $\view^l(v)$ that assigns $1$ only to the root, and $f_i\trianglelefteq f_{i+1}$ for every index $i=2,\ldots,j-1$.
\end{enumerate}
Moreover, we say that $j$ is the \emph{length} of the label $\lab$, denoted by $\len(\lab)$. Let $\labels_j$ be the set of all labels of length at most $j$.

\subsection{The algorithm} \label{subsec:alg:formulation}

In this section we give a high-level description of the algorithm and its pseudo-code formulation.

An important ingredient of the algorithm are meetings between agents during which information is exchanged.
The method that guarantees that some meetings between pairs of agents will occur uses the idea of tunnels introduced in \cite{CLP}.
The routes $R=(e_1,\ldots,e_j)$ and $R'$ form a \emph{tunnel} if $R'=(e_i,e_{i-1},\ldots,e_1,e_1',\ldots,e_{j'}')$ for some $i\in\{1,\ldots,j\}$ and for some $j'\geq 0$.
Moreover, we say that the route $(e_1,\ldots,e_i)$ is the \emph{tunnel core with respect to} $R$.
Note that if $C$ is the tunnel core with respect to $R$, then $\overline{C}$ is the tunnel core with respect to $R'$.
\begin{proposition} \label{prop:tunnel} \textup{(\cite{CLP})}
Let $\agent_i$ be an agent with route $R_i$, $i=1,2$. If $R_1$ and $R_2$ form a tunnel with the tunnel core $C=(e_1,\ldots,e_c)$, then $\agent_1$ and $\agent_2$ are guaranteed to have a meeting such that $(e_1,e_2,\ldots,e_i)$ and $(e_c,e_{c-1},\ldots,e_i)$ are the routes of the agents traversed till the meeting, where $i\in\{1,\ldots,c\}$.
\qed
\end{proposition}
Informally speaking, if the routes of two agents form a tunnel, then they are guaranteed to have a meeting with the property that the routes traversed to date by the agents give (by taking one of the routes and the reversal of the other) the tunnel core.

Let $\cS_n$ be the set of all integer sequences with terms in $\{0,\ldots,n-2\}$, whose length is even and equals at most $6(n-1)$. Then, we define
\begin{multline} \nonumber
\cP^n=\big((\lab,\lab',T)\colon \lab,\lab'\in\labels_{3}\textup{ and }\len(\lab)=\len(\lab')\textup{ and } 
T\in\cS_n\textup{ and }(\lab\neq\lab'\lor(\lab=\lab'\land T=\overline{T}))\big),
\end{multline}
and let $\cP_i^n$ be the $i$-th triple in $\cP^n$, $i=1,\ldots,|\cP^n|$.

In our leader election algorithm we will proceed in phases, and in each phase the label of each agent is fixed.
(Due to the fact that the model is asynchronous, the adversary may force the agents to be in different phases in a particular point of time.)
The total number of phases for each agent is $3$.
After the first phase each agent computes its label used in phase $2$.
These labels are defined in such a way that there exist two agents $\agent$ and $\agent'$ with different labels.
The aim of phase $2$ is that agents $\agent$ and $\agent'$ correctly identify each other's initial positions in their respective views.
After phase $3$ every agent can identify the initial positions of all agents in its view and hence is able to perform leader election.

The label of an agent $\agent$ used in phase $p$ is denoted by $\lab_p(\agent)$, $p=1,2,3$.
Label $\lab_1(\agent)$ is computed before the start of phase $1$, and $\lab_{p+1}(\agent)$ is computed at the end of phase $p$, for $p=1,2,3$.
Label $\lab_4(\agent)$ is used to elect the leader at the end of the algorithm.
Each phase is divided into $|\cP^n|$ stages.
By $R_{p,s}(\agent)$ we denote the route traversed by agent $\agent$ till the end of stage $s$ in phase $p$, $p=1,2,3$, $s=1,\ldots,|\cP^n|$.
As we prove later, each agent $\agent$ starts and ends each stage at its initial position $\home(\agent)$.
Let $R_{p,0}(\agent)$ be the route of an agent $\agent$ traversed till the beginning of phase $p$, and hence till the end of phase $p-1$, whenever $p>0$.
Hence, $R_{1,0}(\agent)$ is the route traversed by $\agent$ prior to the beginning of the first phase.

Now we give an informal description of Algorithm $\algorithmLE$.
This algorithm is executed by each agent, and $\agent$ in the pseudo-code is used to refer to the executing agent.
Note that the upper bound $n$ on the number of nodes of $G$ is given as an input.
The pseudo-code of the algorithm and pseudo-codes of its subroutines are in frames.
In the informal description we refer to lines of these pseudo-codes.

 \newcommand{\lstInitcP}{1}
 \newcommand{\lstInitDFS}{2}
 \newcommand{\lstInitSetf}{3}
 \newcommand{\lstInitLab}{4}

First, we discuss Procedure~$\algorithmInit$ that is called at the beginning of Algorithm $\algorithmLE$.
The agent starts by computing $\cP^n$.
This can be done knowing $n$, without any exploration of the graph.
Then, the agent computes $\view^{\depth}(\home(\agent))$ by performing a DFS traversal of $G$ to the depth $\depth$ (line~$\lstInitDFS$).
The function $f^{\agent}$ is set (line~$\lstInitSetf$) to be the binary mapping for $\view^{\depth}(\home(\agent))$ that assigns $0$ to all nodes of $\view^{\depth}(\home(\agent))$ except for the root.
Hence, $(\view^{\depth}(\home(\agent)),f^{\agent})$ is a partially enhanced view from $\home(\agent)$.
The value of $\lab_1(\agent)$, that will be the label of $\agent$ in the first phase, is set to $(\code(\view^{\depth}(\home(\agent))),f^{\agent})$ (line~$\lstInitLab$).

\begin{center}
\fbox{
\begin{minipage}{0.9\textwidth}
\f{}{0}{\textbf{Procedure} $\algorithmInit(n)$}
\f[35pt]{}{1}{\textbf{Input:} An upper bound $n$ on the size of $G$.}
\f{}{0}{\textbf{begin}}
\f{1:}{1}{Compute $\cP^n$.}
\f{2:}{1}{Compute $\view^{\depth}(\home(\agent))$ by performing a DFS traversal of graph $G$ to depth $\depth$ that ends at $\home(\agent)$.}
\f{3:}{1}{Let $f^{\agent}$ be the binary mapping for $\view^{\depth}(\home(\agent))$ that assigns $1$ only to the root of $\view^{\depth}(\home(\agent))$.}
\f{4:}{1}{$\lab_1(\agent)\leftarrow(\code(\view^{\depth}(\home(\agent))),f^{\agent})$}
\f{}{0}{\textbf{end} $\algorithmInit$}
\end{minipage}
}
\end{center}

 \newcommand{\lstLECallInit}{1}
 \newcommand{\lstLEMainLoopStarts}{2}
 \newcommand{\lstLEInternalLoopStarts}{3}
 \newcommand{\lstLETakeTriple}{4}
 \newcommand{\lstLEFirstCaseFeasible}{5}
 \newcommand{\lstLEFirstCaseRoute}{6}
 \newcommand{\lstLESecondCaseFeasible}{7}
 \newcommand{\lstLESecondCaseRoute}{8}
 \newcommand{\lstLECasesEnd}{9}
 \newcommand{\lstLEInternalLoopEnds}{10}
 \newcommand{\lstLECallUL}{11}
 \newcommand{\lstLENewLabel}{12}
 \newcommand{\lstLEMainLoopEnds}{13}
 \newcommand{\lstLEElectLeader}{14}

Now we informally describe the main part of Algorithm $\algorithmLE$, refering to the lines of the pseudo-code given below.
The $p$-th iteration of the main `for' loop in lines $\lstLEMainLoopStarts$-$\lstLEMainLoopEnds$ is responsible for the traversal performed by the agent in phase $p$, $p\in\{1,2,3\}$.
An internal `for' loop in lines $\lstLEInternalLoopStarts$-$\lstLEInternalLoopEnds$ is executed and its $s$-th iteration determines the behavior of $\agent$ in stage $s$.
The stage $s$ of each phase `processes' the $s$-th element $(\lab',\lab'',T)$ of $\cP^n$.
If $\lab_p(\agent)\notin\{\lab',\lab''\}$, then the agent $\agent$ does not move in this stage and proceeds to the next one.
Otherwise let $\lab_p(\agent)=\lab'$ (we describe only this case as the other one is symmetric).
The agent checks in line~$\lstLEFirstCaseFeasible$ whether a certain trail $(T,H)$ is feasible from $\home(\agent)$, and if it is not, then the stage ends.
As we prove later, the verification of the feasibility of $(T,H)$ can be done by inspecting $\view^{\depth}(\home(\agent))$.
If $(T,H)$ is feasible from $\home(\agent)$, then $\agent$ follows the route $\cR(\home(\agent),(T,H))$ (line~$\lstLEFirstCaseRoute$), which guarantees that:
\begin{itemize}
 \item if the agent $\agent$ is located at $\home(\agent)$ at the beginning of stage $s$, then $\agent$ is located at $\home(\agent)$ at the end of stage $s$ (see Lemma~\ref{lem:closed}), and
 \item if the route $\cR(\home(\agent),T)$ leads from $\home(\agent)$ to the initial position of another agent $\agent'$ and $\lab_p(\agent')=\lab''$, then the routes $R_{p,s}(\agent)$ and $R_{p,s}(\agent')$ form a tunnel.
\end{itemize}
The agent ends phase $p$ by updating its label.
This is done by calling Function $\algorithmUL$ in line~$\lstLECallUL$, which produces a binary mapping $f^{\agent}$ that is used to update the label in line~$\lstLENewLabel$.
After completing the tree phases (for $p=1,2,3$) the agent calls Procedure $\algorithmCL$ which completes the task of leader election.

In order to formally describe the trail $H$ mentioned above we need the following notation.
Let $\lab'\in\labels_3$ be a label of length $p\in \{1,2,3\}$, let $s\in\{1,\ldots,|\cP^n|\}$ and let $v$ be a node of $G$.
We define $\hist(\lab',s)$ to be the trail that corresponds to the route performed till the end of stage $s$ of phase $p$ by an agent $\agent'$ whose label equals $\lab'$ in phase $p$, $\lab_p(\agent')=\lab'$, and whose initial position corresponds to the root of the truncated view $\view^{\depth}(\home(\agent'))$ in $\lab'$.
We prove (see Lemma~\ref{lem:history_simulation}) that, for any $\alpha$ and $s$ the trail $\hist(\lab,s)$ can be computed on the basis of $\lab$ and $s$.
The trail $H$ mentioned above is $\overline{\hist(\lab'',s-1)}$.

\begin{center}
\fbox{
\begin{minipage}{0.9\textwidth}
\f{}{0}{\textbf{Algorithm} $\algorithmLE(n)$}
\f[35pt]{}{1}{\textbf{Input:} An upper bound $n$ on the size of $G$.}
\f{}{0}{\textbf{begin}}
\f{1:}{1}{Call $\algorithmInit(n)$.}
\f{2:}{1}{\textbf{for} $p\leftarrow 1$ \textbf{to} $3$ \textbf{do}} 
\f{3:}{2}{   \textbf{for} $s\leftarrow 1$ \textbf{to} $|\cP^n|$ \textbf{do}}
\f{4:}{3}{      $(\lab',\lab'',T)\leftarrow\cP_s^n$}
\f{5:}{3}{      \textbf{if} $\lab_p(\agent)=\lab'$ and $(T,\overline{\hist(\lab'',s-1)})$ is feasible from $\home(\agent)$ \textbf{then}}
\f{6:}{4}{             Follow the route\\ $\cR(\home(\agent),  (T, \overline{\hist(\lab'',s-1)}, \overline{T}))$.}
\f{7:}{3}{      \textbf{else if} $\lab_p(\agent)=\lab''$ and $(\overline{T}, \overline{\hist(\lab',s-1)})$ is feasible from $\home(\agent)$ \textbf{then}}
\f{8:}{4}{             Follow the route\\ $\cR(\home(\agent),  (\overline{T}, \overline{\hist(\lab',s-1)}, T))$.}
\f{9:}{3}{      \textbf{end if}}
\f{10:}{2}{   \textbf{end for}}
\f{11:}{2}{   $f^{\agent}\leftarrow\algorithmUL(M)$, where $M$ is the memory state of $\agent$.}
\f{12:}{2}{   $\lab_{p+1}(\agent)\leftarrow(\lab_{p},f^{\agent})$}
\f{13:}{1}{\textbf{end for}}
\f{14:}{1}{Call $\algorithmCL$.}
\f{}{0}{\textbf{end} $\algorithmLE$}
\end{minipage}
}
\end{center}

Figure~\ref{fig:tunnel} illustrates the routes of a pair of agents $\agent$ and $\agent'$ that execute one iteration of the internal `for' loop in lines $\lstLEInternalLoopStarts$-$\lstLEInternalLoopEnds$ of $\algorithmLE$, for the same values of $p$ and $s$, such that $(\lab_p(\agent),\lab_p(\agent'),T)=\cP_s^n$.
We assume in this example that $\cR(\home(\agent),T)$ leads from $\home(\agent)$ to $\home(\agent')$ in $G$.
Figure~\ref{fig:tunnel} gives the prefixes of two routes $R_{p,s}(\agent)$ and $R_{p,s}(\agent')$ traversed by the two agents.
The fact that the routes form a tunnel (as shown in Figure~\ref{fig:tunnel}) follows from Lemma~\ref{lem:guaranteed_T_confirmation} proven later.
\begin{figure}[hbt]
	\begin{center}
	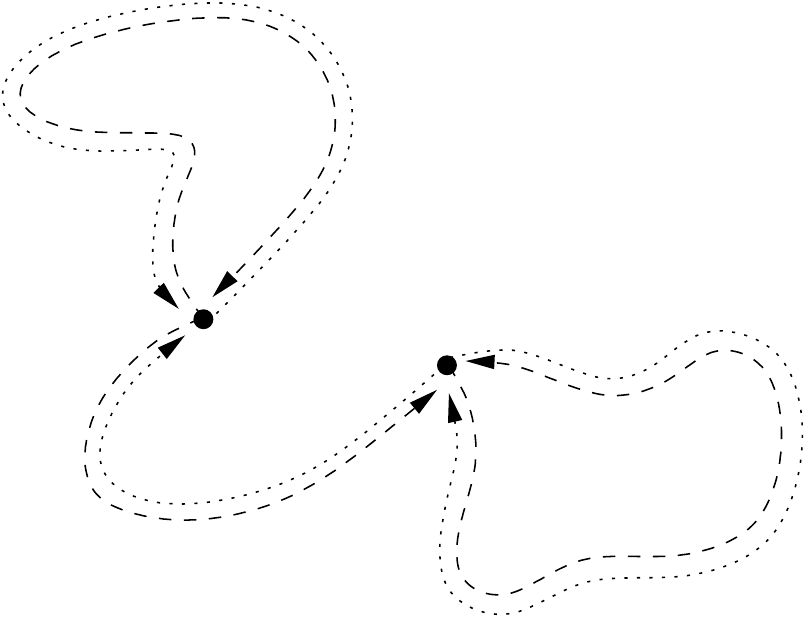
	\caption{The routes $(R_{p,s-1}(\agent),\cR(\home(\agent),  (T, \overline{\hist(\lab'',s-1)})))$ (dashed route) and $(R_{p,s-1}(\agent'),\cR(\home(\agent'),  (\overline{T}, \overline{\hist(\lab',s-1)})))$ (dotted route)}
	\label{fig:tunnel}
	\end{center}
\end{figure}
The routes in Figure~\ref{fig:tunnel} are extended to $R_{p,s}(\agent)$ and $R_{p,s}(\agent')$ once $\agent$ goes from $\home(\agent')$ to its initial position $\home(\agent)$ by following the route $\cR(\home(\agent'),\overline{T})$ in $G$, and $\agent'$ follows $\cR(\home(\agent),T)$ in order to return to $\home(\agent')$, respectively.
Note that the routes $R_{p,s-1}(\agent)$, $R_{p,s-1}(\agent')$ and $\cR(\home(\agent),T)$ are not edge disjoint.

Although we use tunnels, we use them differently from \cite{CLP}.
First, we are able to construct much simpler routes that form tunnels.
This is due to the fact that \cite{CLP} deals with the rendez-vous problem in finite graphs of unknown size and in infinite graphs.
As argued before, for leader election we have to assume that an upper bound $n$ on the size of the graph is known, and we take advantage of knowing $n$ to construct `shorter' tunnels which simplifies our analysis.
Second and more importantly, it is not sufficient for our purposes to just generate a meeting for a particular pair of agents --- the meetings are generated to perform the exchange of information.
In particular, as a result of a meeting that occurs in a tunnel an agent should be able to determine the node (in its own view) corresponding to the initial position of the other agent.
This leads us to the following concept of meetings with `confirmation' of a trail.
\begin{definition} \label{def:confirms}
Let $T\in\cS_n$.
Suppose that agents $\agent$ and $\agent'$ meet.
We say that $\agent$ \emph{confirms} $T$ as a result of this meeting if
\begin{enumerate}[label={\normalfont(\roman*)}]
 \item \label{it:confirms:stage} $\agent$ is in stage $s$ of phase $p$ and $\agent'$ is in stage $s'$ of phase $p'$, where $p'<p$, or $p'=p$ and $s'\leq s$,
 \item \label{it:confirms:label} $\cP_s^n=(\lab',\lab'',T)$, where $\lab_p(\agent)\in\{\lab',\lab''\}$,
 \item \label{it:confirms:trail} if $T_{\agent}$ and $T_{\agent'}$ are the trails traversed by $\agent$ and $\agent'$, respectively, till the meeting, then
  \[(\cT(R_{p,s-1}(\agent)),T', \overline{\hist(\lab,s-1)})=(T_{\agent},\overline{T_{\agent'}}),\]
  where $T'=T$ and $\lab=\lab''$ when $\lab_p(\agent)=\lab'$, and $T'=\overline{T}$ and $\lab=\lab'$ otherwise.
\end{enumerate}
\end{definition}
As we prove in Section~\ref{subsec:alg:analysis}, if an agent $\agent$ confirms $T$ as a result of a meeting with $\agent'$, then $\routeend(\cR(\home(\agent),T'))=\home(\agent')$.

Figure~\ref{fig:verification} depicts the equation in part \ref{it:confirms:trail} of Definition~\ref{def:confirms}. Figure~\ref{fig:verification}(a) presents the trail 
\[(\cT(R_{p,s-1}(\agent)),T', \overline{\hist(\lab,s-1)})\]
that is a perfix of the trail corresponding to the route 
$R_{p,s}(\agent)$  followed by the agent $\agent$ till the end of stage $s$ of phase $p$. Figure~\ref{fig:verification}(b) presents the trails $T_{\agent}$ (dashed line) and $T_{\agent'}$ (dotted line).
\begin{figure}[hbt]
	\begin{center}
	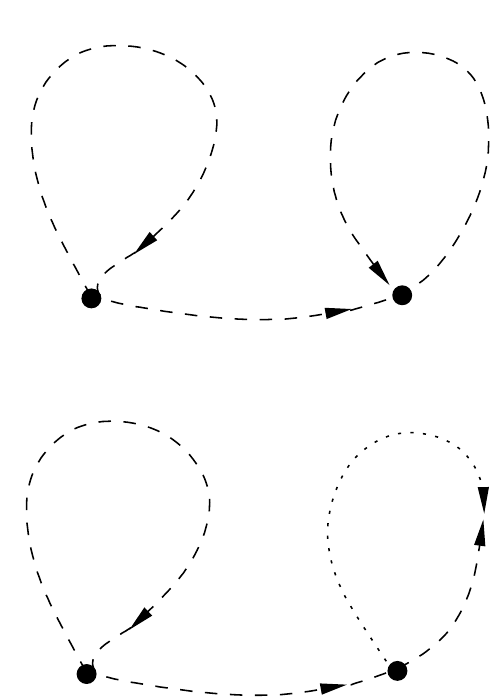
	\caption{(a) the trail $(\cT(R_{p,s-1}(\agent)),T',\overline{\hist(\lab,s-1)})$;
                 (b) a meeting of $\agent$ and $\agent'$ in case when condition \ref{it:confirms:trail} in Definition~\ref{def:confirms} is satisfied}
	\label{fig:verification}
	\end{center}
\end{figure}

It remains to describe Function $\algorithmUL$ and Procedure $\algorithmCL$ that are called in the Algorithm $\algorithmLE$.

 \newcommand{\lstULInitF}{1}
 \newcommand{\lstULMemoryStates}{2}
 \newcommand{\lstULMainLoopStarts}{3}
 \newcommand{\lstULDetermineTrail}{4}
 \newcommand{\lstULIfEssential}{5}
 \newcommand{\lstULNoteOnT}{6}

 \newcommand{\lstULIfShortTrail}{7}
 \newcommand{\lstULFindF}{8}

 \newcommand{\lstULTransition}{9}
 \newcommand{\lstULInternalLoopStarts}{10}
 \newcommand{\lstULUpdateF}{11}
 \newcommand{\lstULIfShortTrailEnds}{12}
 \newcommand{\lstULIfEssentialEnds}{13}
 \newcommand{\lstULMainLoopEnds}{14}
 \newcommand{\lstULReturnF}{15}

\begin{center}
\fbox{
\begin{minipage}{0.95\textwidth}
{\normalsize
\f{}{0}{\textbf{Function} $\algorithmUL(M)$}
\f[35pt]{}{1}{\textbf{Input:} A memory state $M$ of an agent $\agent$.}
\f[35pt]{}{1}{\textbf{Output:} A binary mapping for $\view^{\depth}(\home(\agent))$.}
\f{}{0}{\textbf{begin}}
\f{1:}{1}{Let $f^{\agent}$ be the binary mapping for $\view^{\depth}(\home(\agent))$ that assigns $1$ only to the root of $\view^{\depth}(\home(\agent))$.}
\f{2:}{1}{Let $M_1,\ldots,M_j$ be the memory states of all agents previously met by $\agent$, at the times of the respective meetings.}
\f{3:}{1}{\textbf{for} $i\leftarrow 1$ \textbf{to} $j$ \textbf{do}}
\f{4:}{2}{   Let $T'=T$ if $\lab_p(\agent)=\lab'$, and let $T'=\overline{T}$ if $\lab_p(\agent)=\lab''$, where $(\lab',\lab'',T)=\cP_s^n$.}
\f{5:}{2}{   \textbf{if} $\agent$ or the agent with the memory state $M_i$ confirms $T$ \textbf{then}}
\f{6:}{3}{      $f^{\agent}(x)\leftarrow 1$, where $x$ is the node of $\view^{\depth}(\home(\agent))$ at the end of $T'$ from the root.}
\f{7:}{3}{      \textbf{if} the length of $\cR(\home(\agent),T)$ is at most $n-1$ \textbf{then}}
\f{8:}{4}{      $f'\leftarrow\algorithmUL(M_i)$}
\f{9:}{4}{      Compute transition $\varphi$ from $\view^{2(n-1)}(\home(\agent))$ to $\view^{\depth}(v)$ such that $\varphi(x)$ is the root of $\view^{\depth}(v)$.}
\f{10:}{4}{      \textbf{for each} $y\in\view^{2(n-1)}(v)$ such that $f'(\varphi(y))=1$ \textbf{do}}
\f{11:}{5}{          $f^{\agent}(y)\leftarrow 1$}
\f{12:}{3}{      \textbf{end if}}
\f{13:}{2}{     \textbf{end if}}
\f{14:}{1}{\textbf{end for}}
\f{15:}{1}{\textbf{return} $f^{\agent}$}
\f{}{0}{\textbf{end} $\algorithmUL$}
}
\end{minipage}
}
\end{center}

We start by giving intuition of the first of them.
This procedure is crucial for the entire algorithm, as it takes advantage of memory state exchanges between agents that meet and permits every agent to insert initial positions of all agents in its view.
This in turn allows the agents to learn asymmetries in the initial configuration and thus correctly perform leader election.
Function $\algorithmUL$ takes as an input the current memory state of an agent $\agent$ and returns a binary mapping $f^{\agent}$ for its view $\view^{\depth}(\home(\agent))$, such that $(\view^{\depth}(\home(\agent)),f^{\agent})$ is a partially enhanced view for agent $\agent$.
Agent $\agent$ considers memory states $M_1,\ldots,M_j$ of all previously met agents at the times of the meetings.
The memory state $M_i$, $i\in\{1,\ldots,j\}$, of an agent $\agent'$ and the memory state of $\agent$ at the time of their meeting permit the agent $\agent$ to verify whether $\agent$ or $\agent'$ confirmed $T$ as a result of their meeting.
If one of the agents confirms $T$, then $\agent$ takes the advantage of this fact to determine the nodes of its view corresponding to initial positions of agents.
In particular, $\agent$ is able to locate a node in its own view that corresponds to the initial position of $\agent'$, because there exists a route corresponding to $T$ and connecting the initial positions of the two agents.
Afterwards, if this route is of length at most $n-1$, then $\agent$ recursively calls Function $\algorithmUL$ for the memory state $M_i$ (which is shorter than the current memory state of $\agent$ and thus recursion is correct).
Hence, $\agent$ can compute the binary mappings corresponding to memory states of all previously met agents at the times of the meetings.
Using trails between initial positions of these agents and $\home(\agent)$, as well as the obtained binary mappings, agent $\agent$ can correctly position all partially enhanced views of these agents in its own view.
A call to Function $\algorithmUL$ at the end of phase $p$ executed by agent $\agent$ permits to compute $\lab_{p+1}(\agent)$.

In the formulation of Function $\algorithmUL$ we use the following notions.
Let $u$ and $v$ be two nodes of $G$.
We say that a function $\varphi$ assigning to each node of $\view^{2(n-1)}(u)$ a node of $\view^{3(n-1)}(v)$ is a \emph{transition} from $\view^{2(n-1)}(u)$ to $\view^{3(n-1)}(v)$, if $\varphi(x)$ and $x$ correspond to the same node of $G$ for each node $x$ of  $\view^{2(n-1)}(u)$.
For any trail $T$ and any node $v$, we say that a node $x$ at depth $i$ in $\cV(v)$ is at the end of $T$ from the root, if
the length of $T$ is $2i$ and the sequence of ports corresponding to the simple path from the root of $\cV(v)$ to $x$ is~$T$.

In the following example we illustrate one iteration of the `for' loop in lines~$\lstULMainLoopStarts$-$\lstULMainLoopEnds$ of Function~$\algorithmUL$.
The graph $G$ is given in Figure~\ref{fig:meeting}(a), and let $n=4$ be an upper bound that was initially provided to each agent.
\begin{figure}[htb]
\begin{center}
\includegraphics[scale=0.9]{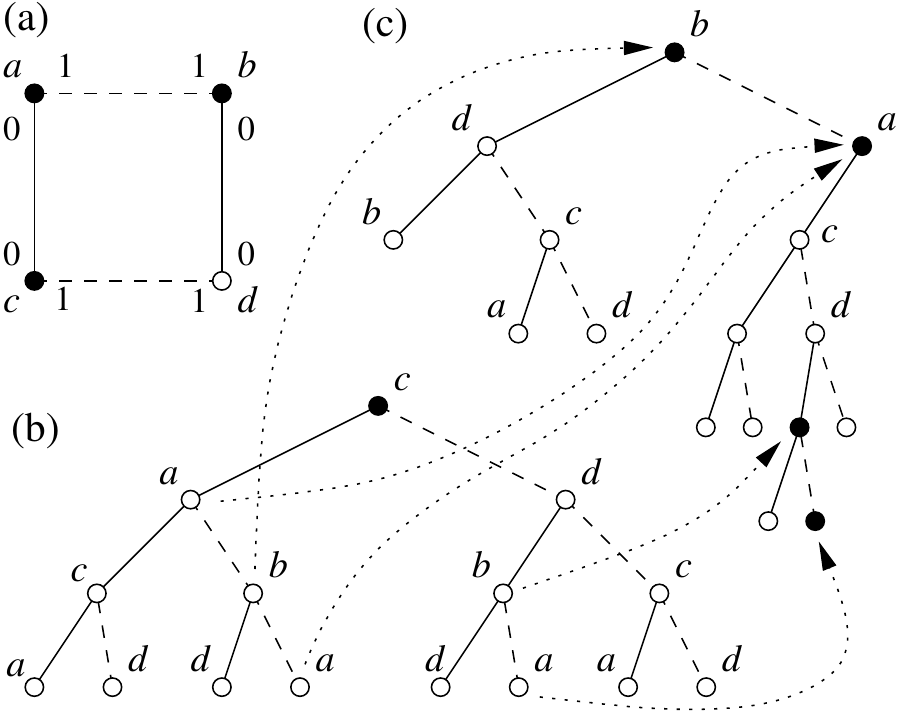}
\caption{(a) a graph $G$; (b) $\view^{n-1}(a)$; a subtree of $\view^{\depth}(b)$}
\label{fig:meeting}
\end{center}
\end{figure}
The black nodes of $G$ are the initial positions of some agents.
Denote by $\agent_a,\agent_b$ and $\agent_c$ the agents whose initial positions are $a,b$ and $c$, respectively.
Note that the views from any two nodes of $G$ are identical in this case.
However, the enhanced view from each node of $G$ is unique.
We focus on the instance of Function~$\algorithmUL$ executed by $\agent_c$ during its meeting with $\agent_b$.
For simplicity, we show only some subtrees of $\view^{\depth}(c)$ and $\view^{\depth}(b)$ in Figures~\ref{fig:meeting}(b) and~\ref{fig:meeting}(c), respectively. Note that the nodes of $G$, and therefore the nodes of any view, are unlabeled and we provide the labels only for the illustrative purpose.
In this example the trail $T'$, computed in line~$\lstULDetermineTrail$ of Function $\algorithmUL$ equals $(0,0,1,1)$, which determines the node of the truncated view $\view^{\depth}(c)$ that corresponds to the root of $\view^{\depth}(b)$.
The black nodes of both views correspond to the initial positions of agents that $\agent_c$ and $\agent_b$ determined prior to this meeting.
The dotted arrows give the part of the transition $\varphi$ that maps the nodes of $\view^{2(n-1)}(c)$ to the black nodes of $\view^{\depth}(b)$.
It follows from the definition of the view that, in general, more than one node of $\view^{n-1}(\home(\agent))$ can be mapped by $\varphi$ to a node of $\view^{\depth}(v)$.

We finally present Procedure $\algorithmCL$ that is called by Algorithm $\algorithmLE$ after the third phase.
The leader is selected by an agent $\agent$ on the basis of the label $\lab_4(\agent)=(\code(\view^{\depth}(\home(\agent))),f_1^{\agent},\ldots,f_4^{\agent})$.
Let $x$ be a node at depth at most $n-1$ in $\view^{\depth}(\home(\agent))$ and satisfying $f_4^{\agent}(x)=1$. Let $S$ be a subtree
of depth $n-1$ of $\view^{\depth}(\home(\agent))$ rooted at $x$. We prove later that  the pair $(S, f_4^{\agent})$ is the complete identifier of some agent.
Since the initial positions of all agents have been detected till the end of phase $3$, the agent can determine all complete identifiers and hence elect the leader.

 \newcommand{\lstCLInitA}{1}
 \newcommand{\lstCLForStarts}{2}
 \newcommand{\lstCLRestrictF}{3}
 \newcommand{\lstCLComputeT}{4}
 \newcommand{\lstCLShouldAdd}{5}
 \newcommand{\lstCLAddToA}{6}
 \newcommand{\lstCLForEnds}{7}
 \newcommand{\lstCLFindMin}{8}
 \newcommand{\lstCLElectLeader}{9}

\begin{center}
\fbox{
\begin{minipage}{0.9\textwidth}
\f{}{0}{\textbf{Procedure} $\algorithmCL$}
\f{}{0}{\textbf{begin}}
\f{1:}{1}{$\cA\leftarrow\emptyset$}
\f{2:}{1}{\textbf{for each} node $x$ of $\view^{n-1}(\home(\agent))$ such that $f_4^{\agent}(x)=1$ \textbf{do}}
\f{3:}{2}{   Let $f'$ be $f_4^{\agent}$ restricted to the nodes of $\view^{n-1}(v)$, where $v$ corresponds to $x$.}
\f{4:}{2}{   Compute trail $T$ such that $v=\routeend(\cR(\home(\agent),T))$.}
\f{5:}{2}{   \textbf{if} $(\code(\view^{n-1}(v)),f')\neq I$ for each $I$ such that $(I,T)\in\cA$ \textbf{then}}
\f{6:}{3}{      $\cA\leftarrow\cA\cup\{((\code(\view^{n-1}(v)),f'),T)\}$}
\f{7:}{1}{\textbf{end for}}
\f{8:}{1}{Find $(I,T)\in\cA$ such that $I=\min\{I'\colon (I',T')\in\cA\}$, where $\min$ is in lexicographic order.}
\f{9:}{1}{Elect the agent whose initial position is $\routeend(\cR(\home(\agent),T))$ to be the leader.}
\f{}{0}{\textbf{end} $\algorithmCL$}
\end{minipage}
}
\end{center}

\subsection{Correctness of the algorithm} \label{subsec:alg:analysis}

This section is devoted to the proof that Algorithm $\algorithmLE$ correctly elects a leader whenever an initial configuration satisfies condition $\conditionEC$, regardless of the actions of the adversary.
The proof is split into a series of lemmas.
The role of the first lemma is to show that knowing $\view^{2n-1}(v)$ is enough to check if a trail of any length is feasible from $v$.

\begin{lemma} \label{lem:view_extension}
Let $v$ be any node of $G$.
Using $\view^{2n-1}(v)$, the truncated view $\view^l(v)$ can be computed for any positive integer $l$.
\end{lemma}
\begin{proof}
If $l\leq 2n-1$, then $\view^l(v)$ is a subtree of $\view^{2n-1}(v)$, so we may assume that $l>2n-1$.
We extend the view $\view^{i}(v)$ to $\view^{i+1}(v)$ for each $i=2n-1,\ldots,l-1$.
To this end we perform the following computation for each node $x$ at depth $i-(n-1)$ of $\view^i(v)$.
Let $u$ be the node of $G$ that corresponds to $x$.
Note that the subtree of $\view^i(v)$ rooted at $x$ and containing all descendants of $x$ in $\view^i(v)$ is equal to $\view^{n-1}(u)$.
Hence, there exists a node $x'$ in $\view^{n-1}(v)$ such that $x'$ corresponds to $u$, in view of the connectedness of $G$.
This implies that there exists a node $y$ in $\view^{n-1}(v)$ such that the subtree of $\view^{2n-1}(v)$ consisting of $y$ and all its descendants to depth $n-1$ from $y$ is equal to $\view^{n-1}(u)$.
Hence, one can find any such node $y$ of $\view^{n-1}(v)$.
Let $u'$ be the node of $G$ that corresponds to $y$.
Due to Proposition~\ref{prop:No}, $\view^n(u')$ is equal to $\view^n(u)$.
Since $y$ belongs to $\view^{n-1}(v)$ we obtain that $\view^n(u')$ is a subtree of $\view^i(v)$ rooted at $y$, and therefore we can extend the subtree $\view^{n-1}(u)$ by replacing in $\view^i(v)$ the subtree rooted at $x$ with $\view^n(u')$.
\end{proof}

\begin{corollary} \label{cor:other_views_computable}
Let $v$ be any node of $G$ and let $x$ be any node of $\view^{2n-1}(v)$.
Then, the subtree of $\view(v)$ to depth $l$ and rooted at $x$ can be computed for any $l>0$, using $\view^{2n-1}(v)$.
\qed
\end{corollary}

\begin{corollary} \label{cor:trail_verification}
Let $T$ be any trail and let $v$ be any node of $G$.
Using $\view^{2n-1}(v)$ it can be verified if $T$ is feasible from $v$ in $G$.
\qed
\end{corollary}

The next lemma shows that given any label $\lab$ of the agent and the stage number $s$, it is possible to compute the trail $\hist(\lab,s)$ which, informally speaking, is the history of the moves of the agent with label $\lab$ till this stage.

\begin{lemma} \label{lem:history_simulation}
Using a label $\lab\in\labels_3$ and an integer $s\in\{1,\ldots,|\cP^n|\}$, the trail $\hist(\lab,s)$ can be computed.
\end{lemma}
\begin{proof}
Suppose that $\lab$ is of length $p$, $p\in\{1,2,3\}$.
By definition, $\lab=(\code(\view^{\depth}(v)),f_1,f_2,\ldots,f_p)$, where $f_i$ is a binary mapping for $\view^{\depth}(v)$ and $v$ is a node of $G$.
First note that $\view^{\depth}(v)$ can be reconstructed from its code.
Suppose that $\agent$ is an agent whose initial position is $v$ and whose label in phase $p$ is $\lab$, $\lab_p(\agent)=\lab$.

The trail  $\hist(\lab,s)$ can be computed by simulating the execution of Algorithm $\algorithmLE$ for the agent $\agent$.
By its formulation, the agent $\agent$ executed Procedure $\algorithmInit$, $p-1$ iterations of the main `for' loop in lines $\lstLEMainLoopStarts$-$\lstLEMainLoopEnds$ of Algorithm $\algorithmLE$,
and exactly $s$ iterations of the nested `for' loop in lines $\lstLEInternalLoopStarts$-$\lstLEInternalLoopEnds$ of Algorithm $\algorithmLE$ in the $p$-th iteration of the main `for' loop. Thus results in traversing the route $\cR(\home(\agent),\hist(\lab,s))$.
Note that $G$ is unknown to $\agent$, but we will reconstruct the route by simulating edge traversals in $\view(v)$.
(While  $\view(v)$ is infinite, it can be reconstructed from $\view^{\depth}(v)$ to any finite depth, using Corollary \ref{cor:trail_verification}.)

We prove the lemma by induction on the total number of stages `processed' by an agent.
Note that if $\lab\in\labels_1$ and $s=0$, then the trail $\hist(\lab,s)$ corresponds to the route that is the DFS traversal of $G$ to depth $3(n-1)$ and starting at $v$.
This trail can be obtained by performing the DFS traversal of $\view^{\depth}(v)$ that starts and ends at the root.

Now assume that $s>0$.
In order to simulate the behavior of $\agent$ in any stage $j$, $j\in\{1,\ldots,|\cP^n|\}$, of phase $i$, $i\in\{1,\ldots,p\}$, one needs to know $\lab_i(\agent)$.
By construction, $\lab_i(\agent)=(\code(\view^{\depth}(v)),f_1,\ldots,f_i)$.
The induction hypothesis and Corollary~\ref{cor:trail_verification} imply that the second part of the condition in line~$\lstLEFirstCaseFeasible$ of Algorithm $\algorithmLE$ can be checked.
If $\lab_i(\agent)\notin\{\lab',\lab''\}$, where $(\lab',\lab'',T)=\cP_j^n$, then $\agent$ does not move in stage $j$ of phase $i$.
Hence, assume without loss of generality that $\lab_i(\agent)=\lab'$.
This implies that $\agent$ executes the instruction in line~$\lstLEFirstCaseRoute$ of Algorithm $\algorithmLE$.
By the induction hypothesis, the trail $\hist(\lab'',j-1)$ in lines $\lstLEFirstCaseFeasible$ and $\lstLEFirstCaseRoute$ can be computed on the basis of $\lab''$ and $j$.
\end{proof}

We say that a route $R$ is \emph{closed} if $\routebegin(R)=\routeend(R)$.

The following lemma implies that at the end of each stage each agent comes back to its initial position.
\begin{lemma} \label{lem:closed}
Let $\agent$ be any agent. For every $p\in\{1,2,3\}$ and for every $s\in\{0,1,\ldots,|\cP^n|\}$ the route $R_{p,s}(\agent)$ is closed.
\end{lemma}
\begin{proof}
Denote by $R_{p,s}'(\agent)$ the route that the agent $\agent$ follows in stage $s$, $s\in\{1,\ldots,|\cP^n|\}$, of phase $p$, $p\in\{1,2,3\}$.
We prove the lemma by induction on the total number of stages processed in all phases by an agent.

First note that $R_{1,0}(\agent)$, i.e., the route of $\agent$ performed as a result of the execution of line~$\lstInitDFS$ of Procedure $\algorithmInit$ is closed.
Hence, it remains to prove that if $R_{p,s}(\agent')$ is closed for each agent $\agent'$, for some $p\in\{1,2,3\}$ and for some $s\in\{0,\ldots,|\cP^n|-1\}$, then $R_{p,s+1}(\agent)$ is closed as well.
Note that $R_{p,s+1}(\agent)=(R_{p,s}(\agent),R_{p,s+1}'(\agent))$.
Hence, $\routebegin(R_{p,s+1}'(\agent))=\routeend(R_{p,s}(\agent))=\home(\agent)$ and therefore it is enough to argue that $R_{p,s+1}'(\agent)$ is closed.

Let $\cP_{s+1}^n=(\lab',\lab'',T)$.
If $\lab_p(\agent)\notin\{\lab',\lab''\}$, then according to lines~$\lstLEFirstCaseFeasible$ and~$\lstLESecondCaseFeasible$ of Algorithm $\algorithmLE$, $R_{p,s+1}'(\agent')$ is empty, i.e. $\agent$ does not move in stage $s+1$ of phase $p$.
In this case the proof is completed.
Otherwise, we obtain that
\[R_{p,s+1}'(\agent)=\cR(\home(\agent),  (T', \overline{\hist(\lab,s)}, \overline{T'})),\]
where $T'\in\{T,\overline{T}\}$ and $\lab\in\{\lab',\lab''\}$ (see lines~$\lstLEFirstCaseRoute$ and~$\lstLESecondCaseRoute$ of Algorithm $\algorithmLE$).
Let $u=\routeend(\cR(\home(\agent),T'))$.
Hence, $R_{p,s+1}'(\agent)$ is closed if and only if $\cR(u,\hist(\lab,s))$ is closed.
However, by definition, the latter route equals $R_{p,s}(\agent')$ for an agent $\agent'$ such that its initial position is $u$ and $\lab_p(\agent')=\lab$ (if such an agent exists).
It follows from the induction hypothesis that $R_{p,s}(\agent')$ is closed, which completes the proof of the lemma.
\end{proof}

%Note that if the routes of two agents form a tunnel, then the agents are guaranteed to meet~\cite{CLP}.
%However, in our algorithm we form a tunnel for an additional important purpose: to verify whether a particular trail $T$ leads in $G$ from the initial position of %one agent to the initial position of another one.
%In other words, the aim of the algorithm is to force meetings that confirm $T$.

The next lemma shows the importance of confirmation of a trail.
It implies that if an agent $\agent$ confirms $T$ as a result of a meeting with $\agent'$, then it can correctly situate the initial position of $\agent'$ in its view.

\begin{lemma} \label{lem:Tconfirmation}
Let $p\in\{1,2,3\}$, let $s\in\{1,\ldots,|\cP^n|\}$ and let $\cP_s^n=(\lab',\lab'',T)$.
Let $\agent$ be an agent such that $\lab_p(\agent)\in\{\lab',\lab''\}$.
Suppose that the agent $\agent$ meets an agent $\agent'$, when $\agent$ is in stage $s$ of phase $p$.
If $\agent$ confirms $T$ as a result of this meeting, then $\home(\agent')=\routeend(\cR(\home(\agent)),T')$, where $T'=T$ if $\lab_p(\agent)=\lab'$ and $T'=\overline{T}$ if $\lab_p(\agent)=\lab''$.
\end{lemma}
\begin{proof}
Suppose without loss of generality that $\lab'=\lab_p(\agent)$.
Condition \ref{it:confirms:trail} in Definition~\ref{def:confirms} implies that the route $R=(R_{p,s-1}(\agent),\cR(\home(\agent),(T',\overline{\hist(\lab'',s-1)})))$ leads from $\home(\agent)$ to $\home(\agent')$ in $G$.
Let $u=\routeend(\cR(\home(\agent),T'))$.
By Lemma~\ref{lem:closed}, both $R_{p,s-1}(\agent)$ and $\cR(u,\overline{\hist(\lab'',s-1)})$ are closed.
This implies that $\home(\agent')=\routeend(R)=\routeend(\cR(\home(\agent),T'))$ as required.
\end{proof}

The following lemma shows that processing an appropriate triple $(\lab,\lab',T)$ by two agents guarantees their meeting confirming $T$.
\begin{lemma} \label{lem:guaranteed_T_confirmation}
Let $p\in\{1,2,3\}$.
Let $\agent$ and $\agent'$ be two agents such that $\cP_s^n=(\lab_p(\agent),\lab_p(\agent'),T)$ for some $T\in\cS_n$ and $s\in\{1,\ldots,|\cP^n|\}$.
If $\routeend(\cR(\home(\agent),T))=\home(\agent')$, then prior to the first moment when one of the agents completes phase $p$, the agents $\agent$ and $\agent'$ have a meeting as a result of which either $\agent$ or $\agent'$ confirms $T$.
\end{lemma}
\begin{proof}
Suppose without loss of generality that $\agent$ ends the traversal of $R_{p,s-1}(\agent)$ at the same time or earlier than $\agent'$ ends the traversal of $R_{p,s-1}(\agent')$.
Let $R_{p,s}'(\agent)$ be the route traversed by $\agent$ in stage $s$ of phase $p$.
By Lemma~\ref{lem:closed}, $\routebegin(R_{p,s}'(\agent))=\routeend(R_{p,s-1}(\agent))=\home(\agent)$.
The route $R_{p,s}'(\agent)$ is constructed as a result of the execution of lines $\lstLEFirstCaseFeasible$-$\lstLEFirstCaseRoute$ of Algorithm $\algorithmLE$ by $\agent$.
Since $\lab''=\lab_p(\agent')$ in line $\lstLEFirstCaseRoute$ of Algorithm $\algorithmLE$, we obtain that $R_{p,s}'(\agent)$ and $R_{p,s}(\agent')$ form a tunnel with the tunnel core $C=(\cR(\home(\agent),(T,\overline{\hist(\lab_p(\agent'),s-1)})))$.
By Proposition~\ref{prop:tunnel}, $\agent$ and $\agent'$ will have a meeting while $\agent$ is in stage $s$ of phase $p$ and before $\agent'$ ends the traversal of $R_{p,s}(\agent')$.
This implies that \ref{it:confirms:stage} of Definition~\ref{def:confirms} holds.
Moreover, \ref{it:confirms:label} of Definition~\ref{def:confirms} is satisfied by assumption.
Let $T_{\agent}$ and $T_{\agent'}$ be the trails traversed by $\agent$ and $\agent'$, respectively, till the meeting.
By Proposition~\ref{prop:tunnel}, $(\cT(R_{p,s-1}(\agent)),\cT(C))=(T_{\agent},\overline{T_{\agent'}})$, where $C$ is the tunnel core.
This proves that \ref{it:confirms:trail} of Definition~\ref{def:confirms} holds.
Hence, the agent $\agent$ confirms $T$ as a result of the meeting.
\end{proof}

The role of the next lemma is to show that an agent never marks falsely an initial position of another agent in its view.
\begin{lemma} \label{lem:no_false_noting}
Let $\lab_p(\agent)=(\code(\view^{\depth}(\home(\agent))),f_1^{\agent},\ldots,f_p^{\agent})$ be the label of any agent $\agent$ in phase $p\in\{1,2,3\}$.
If $x$ is any node of $\view^{\depth}(\home(\agent))$ corresponding to a node of $G$ that is not an initial position of an agent, then $f_p^{\agent}(x)=0$.
\end{lemma}
\begin{proof}
Suppose for a contradiction that an agent $\agent$ sets $f^{\agent}(x)$ to be $1$, and $x$ corresponds to a node of $G$ that is not an initial position of an agent.

Suppose that the input memory state $M$ of the agent $\agent$ is the shortest that satisfies this property.
This assumption implies that if $f'$ is computed by agent $\agent$ in line~$\lstULFindF$ of Function $\algorithmUL$ for any memory state $M_i$, $i=1,\ldots,j$, then $(\view^{\depth}(v),f')$ is a partially enhanced view from $v$, where $v$ corresponds to the node $x$ at the end of $T'$ from the root in $\view^{\depth}(\home(\agent))$.
Hence, if $v$ is an initial position of an agent, then, due to the definition of transition, each node $y$ from line~$\lstULUpdateF$ of Function $\algorithmUL$ corresponds to an initial position of an agent.

The latter implies that the assignment of $1$ to $f^{\agent}(x)$ occurs in line~$\lstULNoteOnT$ of $\algorithmUL$, as a consequence of a meeting with some agent $\agent'$ whose memory state was $M_i$ at the time of the meeting.
Hence, $x$ is at the end of $T'$ from the root in $\view^{\depth}(\home(\agent))$.
Due to line $\lstULIfEssential$ of $\algorithmUL$, either $\agent$ or $\agent'$ confirms $T$ as a result of their meeting.
By Lemma~\ref{lem:Tconfirmation}, $\routeend(\home(\agent),T')=\home(\agent')$, where $T'\in\{T,\overline{T}\}$ is determined in line~$\lstULDetermineTrail$ of $\algorithmUL$.
Thus, $\home(\agent')$ corresponds to the node at the end of $T'$ from the root in $\view^{\depth}(\home(\agent))$.
The latter implies that $\home(\agent')$ corresponds to $x$, contradicting our assumption.
\end{proof}

The next lemma is a companion result to Lemma~\ref{lem:no_false_noting}.
It says that if an agent confirms a trail $T$ as a result of a meeting with $\agent'$, then both of them correctly mark their respective initial positions in their views.
\begin{lemma} \label{lem:updateF}
Let $p\in\{1,2,3\}$.
Let $\agent$ and $\agent'$ be two agents with labels $\lab_{p+1}(\agent)=(\lab_p(\agent),f_{p+1}^{\agent})$ and $\lab_{p+1}(\agent')=(\lab_p(\agent'),f_{p+1}^{\agent'})$.
If agent $\agent$ confirms $T$ as a result of a meeting with agent $\agent'$ in phase $p'$, $p'\leq p$, then
\begin{enumerate}[label={\normalfont(\roman*)}]
 \item $f_{p+1}^{\agent}(x)=1$, where $x$ is at the end of $T'$ from the root in $\view^{\depth}(\home(\agent))$, such that $T'\in\{T,\overline{T}\}$ and $\routeend(\cR(\home(\agent),T'))=\home(\agent')$.
 \item $f_{p+1}^{\agent'}(x')=1$, where $x'$ is at the end of $\overline{T'}$ from the root in $\view^{\depth}(\home(\agent'))$, and $\routeend(\cR(\home(\agent'),\overline{T'}))=\home(\agent)$.
\end{enumerate}
\end{lemma}
\begin{proof}
Suppose that $\agent$ is in some stage of phase $p'$, $p'\leq p$, when a meeting with $\agent'$ occurs as a result of which $\agent$ confirms $T$.
Hence, one of the memory states in line~$\lstULMemoryStates$ of Function $\algorithmUL$ called at the end of phase $p$ is the memory state $M_i$, $i\in\{1,\ldots,j\}$, of $\agent'$ at the time of the meeting.
We consider the $i$-th iteration of the `for' loop in lines $\lstULMainLoopStarts$-$\lstULMainLoopEnds$ of Function $\algorithmUL$, i.e., informally speaking, the iteration in which $\agent$ `analyzes' the meeting with $\agent'$.
The agent $\agent$ determines in line~$\lstULIfEssential$ of Function $\algorithmUL$ the fact that either $\agent$ or $\agent'$ confirms $T$ as a result of the meeting.
Then, in line~$\lstULNoteOnT$ of Function $\algorithmUL$, $f^{\agent}(x)$ is set to $1$, where $x$ is at the end of $T'$ from the root of $\view^{\depth}(\home(\agent))$.
By Lemma~\ref{lem:Tconfirmation} and by the choice of $T'$ in line~$\lstULDetermineTrail$ of Function $\algorithmUL$, $x$ corresponds to $\home(\agent')$.
Function $\algorithmUL$ returns $f^{\agent}$.
Due to line~$\lstLENewLabel$ of Algorithm $\algorithmLE$, $f_{p+1}^{\agent}(x)=1$.

In order to prove the second part of the lemma notice that agent $\agent'$ has also access to the memory states of $\agent$ and $\agent'$ at the time of the meeting.
Hence, upon completing phase $p$ it performs analogous computations as $\agent$ during its execution of Function $\algorithmUL$ at the end of phase $p$ and sets $f_{p+1}^{\agent'}(x')=1$.
\end{proof}

Let $\agent$ be an agent and let $x$ be a node of the truncated view $\view^{\depth}(\home(\agent))$ at the end of a trail $T$ from the root of $\view^{\depth}(\home(\agent))$.
If the agent $\agent$ sets $f^{\agent}(x)=1$ during the execution of Function $\algorithmUL$ and $x$ corresponds to the initial position of some agent $\agent'$, then we say that $\agent$ \emph{noted $\agent'$ on $T$}.

A label $\lab=(\code(\view^{\depth}(\home(\agent))),f_1^{\agent},\ldots,f_p^{\agent})$, where $p\in\{1,2,3\}$, of an agent $\agent$ is \emph{complete with respect to} $\agent'$ if $f_p^{\agent}(x)=1$ for each node $x$ of $\view^{\depth}(\home(\agent))$ that corresponds to an initial position of $\agent'$.
We say that $\lab$ is \emph{semi-complete with respect to} $\agent'$ if $f_p^{\agent}(x)=1$ for each node $x$ that corresponds to an initial position of $\agent'$ and belongs to a level $i\leq 2(n-1)$  of $\view^{\depth}(\agent)$.

The next lemma explains how agents confirm trails as a result of their meetings.
Agents with different labels can confirm any trail in $\cS_n$ between their initial positions, while agents with equal labels are able to confirm only `palindromes'.
\begin{lemma} \label{lem:label_complete}
Let $p\in\{1,2,3\}$ and let $\agent$ and $\agent'$ be any two agents.
\begin{enumerate}[label={\normalfont(\roman*)}]
 \item\label{it:different_labels} If $\lab_p(\agent)\neq\lab_p(\agent')$, then $\lab_{p+1}(\agent)$ is complete with respect to $\agent'$.
 \item\label{it:identical_labels} If $\lab_p(\agent)=\lab_p(\agent')$, then prior to the end of its phase $p$ the agent $\agent$ has noted $\agent'$ on each trail $T$ such that $T\in\cS_n$,$T=\overline{T}$ and $\home(\agent')=\routeend(\cR(\home(\agent),T))$.
\end{enumerate}
\end{lemma}
\begin{proof}
To prove \ref{it:different_labels} let $T\in\cS_n$ be any trail such that $\cR(\home(\agent),T)$ leads from $\home(\agent)$ to $\home(\agent')$ in $G$.
By construction of $\cP^n$, $\cP_i^n=(\lab_p(\agent),\lab_p(\agent'),T)$ for some index $i\in\{1,\ldots,|\cP^n|\}$, because $\lab_p(\agent)\neq\lab_p(\agent')$.
By Lemma~\ref{lem:guaranteed_T_confirmation}, one of the agents $\agent$ or $\agent'$ confirms $T$ as a result of a meeting that occurs when it is in phase $p$.
This is done by checking the conditions \ref{it:confirms:stage}, \ref{it:confirms:label} and \ref{it:confirms:trail} in Definition~\ref{def:confirms}, which can be accomplished by analyzing the memory states of $\agent$  and $\agent'$ at the time of the meeting.
By Lemma~\ref{lem:updateF}, $f^{\agent}(x)=1$, where $x$ is at the end of $T$ in $\view^{\depth}(\home(\agent))$, regardless of which agent confirmed $T$ as a result of the meeting.

The proof of \ref{it:identical_labels} is analogous.
\end{proof}

\begin{lemma} \label{lem:phase1}
There exist two agents $\agent$ and $\agent'$ such that $\lab_2(\agent)\neq\lab_2(\agent')$.
\end{lemma}
\begin{proof}
If there exist two agents $\agent$ and $\agent'$ with different views, $\view(\home(\agent))\neq\view(\home(\agent'))$, then by Proposition~\ref{prop:No}, $\view^{\depth}(\home(\agent))\neq\view^{\depth}(\home(\agent'))$.
Proposition~\ref{prop:codes_distinguish} implies $\lab_1(\agent)\neq\lab_1(\agent')$, and consequently $\lab_2(\agent)\neq\lab_2(\agent')$.
Thus, assume in the following that the views from the initial positions of all agents are equal.
In view of condition $\conditionEC$ there exists a non-uniform palindrome.
Hence, there exist two agents $\agent,\agent'$ and a trail $T$, such that $T=\overline{T}$, $\cR(\home(\agent),T)$ leads from $\home(\agent)$ to a node $u$ that is an initial position of an agent $\agent''$, and $\cR(\home(\agent'),T)$ leads from $\home(\agent')$ to a node $u'$ that is not an initial position of any agent.
Let $x$ and $x'$ be the nodes of $\view^{\depth}(\home(\agent))$ and $\view^{\depth}(\home(\agent'))$, respectively, at the end of $T$ from the roots.
Note that $x$ corresponds to $u$ and $x'$ corresponds to $u'$.
By Lemma~\ref{lem:label_complete}\ref{it:identical_labels}, $\agent$ noted $\agent''$ on $T$ in phase $1$ and, by Lemma~\ref{lem:no_false_noting}, $\agent'$ does not note any agent on $T$ during the entire execution of its algorithm.
Hence, $f^{\agent}(x)=1$ and $f^{\agent'}(x')=0$ after the execution of Procedure $\algorithmUL$ at the end of phase $1$.
Since the views  $\view^{\depth}(\home(\agent))$ and $\view^{\depth}(\home(\agent'))$ are equal, we have
 $\lab_2(\agent)=(\code(\view^{\depth}(\home(\agent))),f_1^{\agent},f_2^{\agent})$ and
 $\lab_2(\agent')=(\code(\view^{\depth}(\home(\agent'))),f_1^{\agent'},f_2^{\agent'})$, where $f_2^{\agent}\neq f_2^{\agent'}$.
Hence, $\lab_2(\agent)\neq\lab_2(\agent')$.
\end{proof}

The role of the next lemma is to explain indirect learning of initial positions of other agents.
If two agents $\agent$ and $\agent'$ have equal labels, then they mutually situate their initial positions using a third agent $\agent''$ with a different label as an intermediary.
Such an agent exists by Lemma~\ref{lem:phase1}.
Agents $\agent$ and $\agent'$ can correctly fill their initial positions to depth $2(n-1)$ in their views due to the fact that the intermediary $\agent''$ has their initial positions to depth $\depth$.

\begin{lemma} \label{lem:completeTOsemicomplete}
Let $\agent$ and $\agent'$ be two agents such that $\lab_p(\agent)$ is complete with respect to $\agent'$, $p\in\{1,2,3\}$.
If $\agent''$ is any agent such that $\lab_p(\agent'')\neq\lab_p(\agent)$, then $\lab_{p+1}(\agent'')$ is semi-complete with respect to $\agent'$.
\end{lemma}
\begin{proof}
Since $\lab_p(\agent'')\neq\lab_p(\agent)$ we obtain, by the definition of $\cP^n$, that $(\lab_{p}(\agent''),\lab_{p}(\agent),T)=\cP_i^n$ for some $i\in\{1,\ldots,|\cP^n|\}$, where $T$ is such a trail that the route $R''=\cR(\home(\agent''),T)$ contains at most $n-1$ edges and $\routeend(R'')=\home(\agent)$.
We analyze the execution of Function $\algorithmUL$ by $\agent''$ at the end of phase $p$.
By Lemma~\ref{lem:guaranteed_T_confirmation}, $\agent$ or $\agent''$ confirms $T$ as a result of their meeting prior to the completion of phase $p$ by $\agent''$, i.e., prior to the execution of Function $\algorithmUL$ we consider.
The agent $\agent''$ learns this fact in line~$\lstULIfEssential$ of Function $\algorithmUL$.
Then, the condition in line~$\lstULIfShortTrail$ of Function $\algorithmUL$ is satisfied by assumption, and $\agent''$ verifies this condition by checking whether the length of $T$ does not exceed $2(n-1)$.
Consequently, $\agent''$ finds in line~$\lstULTransition$ of Function $\algorithmUL$ a transition $\varphi$ that maps the node at the end of $T$ in $\view^{\depth}(\home(\agent''))$ to the root of $\view^{\depth}(\home(\agent))$, because $\routeend(R'')=\home(\agent)$.
Since $R''$ is of length at most $n-1$, $\agent''$ sets $f^{\agent''}(y)=1$ (in the `for' loop in lines $\lstULInternalLoopStarts$-$\lstULUpdateF$ of Function $\algorithmUL$) for each node $y$ at depth at most $2(n-1)$ in $\view^{\depth}(\home(\agent''))$ corresponding to $\home(\agent')$.
The latter is due to the fact that $\lab_p(\agent)$ is complete with respect to $\agent'$.
This proves that $\lab_{p+1}(\agent'')$ is semi-complete with respect to $\agent'$.

See Figure~\ref{fig:complete}, where the root and some nodes that correspond to the initial position of $\agent'$ have been marked on $\view^{3(n-1)}(\home(\agent))$, and we show a transition that allows to determine all nodes in the view of $\agent''$ to the depth $2(n-1)$ corresponding to $\home(\agent')$.
\begin{figure}[hbt]
	\begin{center}
	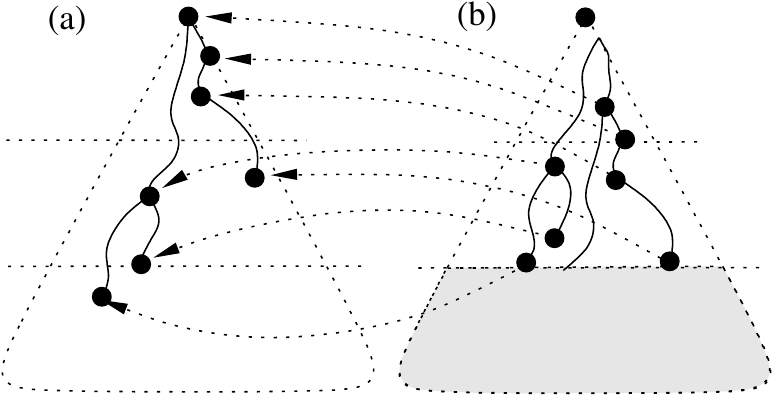
	\caption{(a) $\view^{\depth}(\home(\agent))$; (b) $\view^{2(n-1)}(\home(\agent''))$; $\agent''$ learns all nodes corresponding to $\home(\agent')$ to the depth $2(n-1)$ in its view, on the basis of $\view^{\depth}(\home(\agent))$}
	\label{fig:complete}
	\end{center}
\end{figure}
\end{proof}

In view of the next lemma, upon completion of phase $3$ an agent has correctly situated all nodes corresponding to initial positions of all agents to depth $2(n-1)$ in its view.

\begin{lemma} \label{lem:phase3}
Let $\agent$ be any agent.
Then, $\lab_4(\agent)$ is semi-complete with respect to each agent $\agent'$.
\end{lemma}
\begin{proof}
If $\lab_3(\agent)\neq\lab_3(\agent')$, then by Lemma~\ref{lem:label_complete}\ref{it:different_labels}, $\lab_4(\agent)$ is complete with respect to $\agent'$, and therefore it is semi-complete with respect to $\agent'$.
Hence, suppose that $\lab_3(\agent)=\lab_3(\agent')$, which implies $\lab_2(\agent)=\lab_2(\agent')$.
Hence, Lemma~\ref{lem:phase1} implies that there exists an agent $\agent''$ such that $\lab_2(\agent'')\neq\lab_2(\agent)$.
Due to Lemma~\ref{lem:label_complete}\ref{it:different_labels}, $\lab_3(\agent'')$ is complete with respect to $\agent'$.
Note that $\lab_2(\agent)\neq\lab_2(\agent'')$ implies $\lab_3(\agent)\neq\lab_3(\agent'')$.
Lemma~\ref{lem:completeTOsemicomplete} completes the proof.
\end{proof}

Our final lemma says that any agent can correctly reconstruct complete identifiers of all other agents.
This is due to the fact that every agent has filled initial positions of all agents to depth $2(n-1)$ in its view.

\begin{lemma} \label{lem:complete_identifiers}
Let $\cA=\{(I_1,T_1),\ldots,(I_a,T_a)\}$ be the set computed in Procedure $\algorithmCL$ executed by an $\agent$.
Then, for every agent $\agent'$, there exists an index $i\in\{1,\ldots,a\}$ such that $I_i$ is the complete identifier of $\agent'$ and $\cR(\home(\agent),T_i)$ leads from $\home(\agent)$ to $\home(\agent')$ in $G$.
\end{lemma}
\begin{proof}
Let $\lab_4(\agent)=(\code(\view^{\depth}(\home(\agent))),f_1^{\agent},\ldots,f_4^{\agent})$ be the label of $\agent$ obtained at the end of phase $3$.
Initially $\cA$ is set to be empty in line~$\lstCLInitA$ of Procedure $\algorithmCL$.
Consider the nodes $v$ and $x$ and the function $f'$ from line~$\lstCLRestrictF$ of Procedure $\algorithmCL$.
By Lemma~\ref{lem:no_false_noting}, $v$ is an initial position of an agent $\agent''$, because $f'(x)=1$.
Thus, in view of Lemma~\ref{lem:phase3}, $(\code(\view^{n-1}(v)),f')$ is the complete identifier of $\agent''$.
Due to lines $\lstCLComputeT$-$\lstCLAddToA$ of Procedure $\algorithmCL$, $((\code(\view^{n-1}(v)),f'),T)\in\cA$ for some trail $T$ such that  $\routeend(\cR(\home(\agent),T))=v=\home(\agent'')$.
By  Lemma~\ref{lem:phase3}, for any agent $\agent'$, there is a node in $\view^{n-1}(\home(\agent))$ corresponding to $\home(\agent')$.
This implies that, for any agent $\agent'$, the set $\cA$ contains (at the end of the `for' loop in lines $\lstCLForStarts$-$\lstCLForEnds$ of Procedure $\algorithmCL$) a pair $(I_i,T_i)$, $i\in\{1,\ldots,a\}$, such that $I_i$ is the complete identifier of $\agent'$ and $\cR(\home(\agent),T_i)$ leads from $\home(\agent)$ to $\home(\agent')$ in $G$.
\end{proof}

\begin{theorem} \label{thm:LE_part2}
If the condition $\conditionEC$ is satisfied for an initial configuration, then Algorithm $\algorithmLE$ correctly elects a leader regardless of the actions of the adversary.
\end{theorem}
\begin{proof}
Let $\agent^*$ be the agent whose complete identifier $I^*$ is lexicographically smallest among the complete identifiers of all agents.
By condition $\conditionEC$ (more precisely by its part saying that all enhanced views are different) and by Corollary~\ref{cor:No2}, agent $\agent^*$ is unique.

Each agent $\agent$ computes $\lab_4(\agent)$ as a result of the execution of Algorithm $\algorithmLE$.
By Lemma~\ref{lem:complete_identifiers}, $\agent$ computes (in lines $\lstCLInitA$-$\lstCLForEnds$ of Procedure $\algorithmCL$), for each agent $\agent'$, the complete identifier of $\agent'$ and a trail $T$ such that the route $\cR(\home(\agent),T)$ of length at most $n-1$ leads from $\home(\agent)$ to $\home(\agent')$ in $G$.
By Corollary~\ref{cor:No2}, the complete identifiers uniquely distinguish the agents.
Using the lexicographic order $\preceq$ on the set of all complete identifiers, the agent $\agent$ finds in line $\lstCLFindMin$ of Procedure $\algorithmCL$ the complete identifier $I=(\view^{n-1}(u),f)$ such that $(I,T)\in\cA$ and $I\preceq I'$ for each $I'$ such that $(I',T')\in\cA$ for some trail $T'$.
By definition, $I=I^*$.
Then, $\agent$ decides in line~$\lstCLElectLeader$ of Procedure $\algorithmCL$ that the agent with the initial position $\routeend(\cR(\home(\agent),T))$ is the leader.
This is agent $\agent^*$.
\end{proof}

Theorem~\ref{thm:LE_part2}, together with Proposition~\ref{pro:negative1} implies our main result which is Theorem~\ref{main} from Section~\ref{sec:feasibility}.

\section{Conclusion} \label{sec:conclusions}

We characterized all initial configurations of agents for which leader election is possible and we constructed a universal algorithm electing a leader
for all such configurations, assuming that agents know an upper bound on the size of the graph. We observed that the latter assumption cannot be removed.
In this paper we focused on the feasibility of leader election under a very harsh scenario in which the adversary  
controls the speed and the way in which agents move along their chosen routes. This adversarial scenario captures the totally asynchronous nature of mobile agents.

While we gave a complete solution to the problem of feasibility of leader election, we did not try to optimize the efficiency of the algorithm, e.g., in terms of
its cost, i.e., of the total or of the maximum number of edge traversals performed by the mobile agents. 
In fact, any kind of such optimization appears to be quite challenging.
It is clear that in order to elect a leader agents have to meet. Already the much simpler  problem of optimizing the cost of meeting of two agents in our asynchronous model
is open, both when agents have different labels \cite{CLP} and when they are anonymous, as in our present scenario \cite{GP}. In particular, in the latter
paper the authors asked if rendezvous of two agents can be accomplished (whenever it is feasible) at a cost  polynomial in the size of the graph.

%%%%%%%%%%%%%%%%%%%%%%%%%%%%%%%%%%%%
\bibliographystyle{plain}

% \begin{thebibliography}{}
% %
% % and use \bibitem to create references. Consult the Instructions
% % for authors for reference list style.
% %
% \bibitem{RefJ}
% % Format for Journal Reference
% Author, Article title, Journal, Volume, page numbers (year)
% % Format for books
% \bibitem{RefB}
% Author, Book title, page numbers. Publisher, place (year)
% % etc
% \end{thebibliography}

\end{document}